\documentclass[review]{elsarticle}

\usepackage{a4}
\usepackage{amsmath}
\usepackage{amsthm}
\usepackage{amssymb}
\usepackage[mathscr]{eucal}
\usepackage[]{color}
\usepackage{tabularx}
\newcolumntype{Y}{>{\centering\arraybackslash}X}
\usepackage{rotating}
\usepackage{multirow}
\usepackage{listings}
\usepackage{graphicx}
\graphicspath{{images/}}
\usepackage{caption}
\usepackage{subcaption}
\usepackage{mwe}
\usepackage{tikz}
\usepackage[ruled]{algorithm2e}
\usepackage{dsfont}
\usepackage{bbm}
\usepackage{bm}
\usepackage{url,xcolor}
\definecolor{hotpink}{rgb}{0.9,0,0.5}
\newcommand{\field}[1]{\mathbb{#1}}
\newcommand{\R}{\field{R}}
\newcommand{\N}{\field{N}}
\newcommand{\Exp}{\field{E}}
\newcommand{\Pro}{\field{P}}
\newcommand{\df}{\,\mathrm{d}}

\newtheorem{theorem}{Theorem}[section]
\newtheorem{lemma}{Lemma}[section]

\newtheorem{proposition}{Proposition}[section]

\newtheorem{assumption}{Assumption}[section]

\DeclareMathOperator{\regret}{Regret}

\DeclareMathOperator{\Var}{Var}
\DeclareMathOperator{\tr}{tr}

\DeclareMathOperator*{\argmax}{arg\,max}

\usepackage{ulem}

\numberwithin{equation}{section}
\definecolor{keywordcolor}{rgb}{0.8,0.1,0.5}
\usepackage{lineno,hyperref}
\modulolinenumbers[5]

\journal{Journal of \LaTeX\ Templates}

\begin{document}

\begin{frontmatter}

\title{Adaptive Pricing in Insurance: Generalized Linear Models and Gaussian Process Regression Approaches}

\author[address]{Yuqing Zhang \corref{cor1}}
\ead{yuqing.zhang@manchester.ac.uk}
%\ead[url]{http://www.maths.manchester.ac.uk/~yuqing/}

\author[address]{Neil Walton\corref{cor2}}
\ead{neil.walton@manchester.ac.uk}
\ead[url]{https://sites.google.com/site/neilwaltonswebsite/}

\cortext[cor1]{Corresponding author}
\cortext[cor2]{Principal corresponding author}

\address[address]{School of Mathematics, The University of Manchester, Manchester, M13 9PL, UK.}
%\fntext[myfootnote]{Since 2016.}

%\author[mymainaddress]{Neil Walton \corref{mycorrespondingauthor}}
%\cortext[mycorrespondingauthor]{Corresponding authors}
%\ead{neil.walton@manchester.ac.uk}
%\ead[url]{https://sites.google.com/site/neilwaltonswebsite/}

%\address[mymainaddress]{School of Mathematics, The University of Manchester, Manchester, M13 9PL, UK.}

\begin{abstract}
We study the application of dynamic pricing to insurance.
We view this as an online revenue management problem where the insurance company looks to set prices to optimize the long-run revenue from selling a new insurance product.
We develop two pricing models:
an adaptive Generalized Linear Model (GLM) and an adaptive Gaussian Process (GP) regression model.
Both balance between exploration, where we choose prices in order to learn the distribution of demands \& claims for the insurance product, and exploitation, where we myopically choose the best price from the information gathered so far.
The performance of the pricing policies is measured in terms of regret: the expected revenue loss caused by not using the optimal price.
As is commonplace in insurance, we model demand and claims by GLMs.
In our adaptive GLM design, we use the maximum quasi-likelihood estimation (MQLE) to estimate the unknown parameters.
We show that, if prices are chosen with suitably decreasing variability, the MQLE parameters eventually exist and converge to the correct values, which in turn implies that the sequence of chosen prices will also converge to the optimal price.
In the adaptive GP regression model, we sample demand and claims from Gaussian Processes and then choose selling prices by the upper confidence bound rule.
We also analyze these GLM and GP pricing algorithms with delayed claims.
Although similar results exist in other domains, this is among the first works to consider dynamic pricing problems in the field of insurance. We also believe this is the first work to consider Gaussian Process regression in the context of insurance pricing. These initial findings suggest that online machine learning algorithms could be a fruitful area of future investigation and application in insurance.

\end{abstract}


\end{frontmatter}

%%%%%%%%%%%%%%%%%%%%%%%%%%%%%%%%%%%%%%%%%%%%%
%% Introduction %%
%%%%%%%%%%%%%%%%%%%%%%%%%%%%%%%%%%%%%%%%%%%%%
\section{Introduction}

We study the application of dynamic pricing from the perspective of an insurance company.
Here the insurance company looks to set prices to optimize the long-run revenue from selling insurance products which also experience a distribution of claims.
This can be cast as a revenue management problem, see Philips \cite{PhillipsRobert2005} and Talluri and Ryzin \cite{TalluriRyzin2005} for an overview of this field.
If the distribution of demand and claims was known to the insurance company, then this could be formulated as a relatively straight-forward optimization problem.
However, in the real world, demand for a new product at each price is not deterministic and known.
Thus, we assume that the insurance company only observes the realised demand and does not know the underlying distribution of demand and claims for the insurance product.
This is particular relevant for the release of new insurance products.

Given demand and claims are not known, the retailer faces a learning \& pricing problem.
This is sometimes known as a exploration-exploitation trade-off.
At the beginning of each selling period, we set a price close to the estimated best price
and then study the changes of demand and claims when varying prices.
This is the exploration process, which enables us to find the relationship between price and demand \& claims distributions.
Further by setting the price close to the estimated best price, we are able to exploit what we have learned.
This is the exploitation process.
Choosing prices that are far from the best estimated price encourages exploration but can be inefficient in exploiting available information.
On the other hand, choosing pricing close to the best estimate price may not discover enough about the underlying distribution of claims to converge to the optimal price.
Therefore, the insurance company must create a policy for pricing the insurance product that reveals sufficient information about the underlying demand and claims distributions so as to optimize the long-run revenue of the insurance company.
The policy that we consider provides a mechanism for efficiently exploring different prices offering, and then exploiting that knowledge to achieve the revenue maximization objective.

We consider this pricing problem as a multi-armed bandit problem, which has been widely used to address the trade-off between exploration and exploitation in sequential decision making.
We investigate two regression models for the learning \& pricing problem: the Generalized Linear Models (GLM) and the Gaussian Process (GP).
GLMs are a classical statistical technique, introduced by Nelder and Wedderburn \cite{NelderWedderburn1972}, and was first applied in  insurance rating by
McCullagh and Nedler \cite{McCullaghNelder1989}.
However, for sequential decision making, prices derived from maximum likelihood estimates may not be consistent due to insufficient exploration \cite{LaiRobbins1979}.
To solve this problem, the strong consistency of least squares estimates is required, which is established by Lai and Robbins \cite{LaiRobbins1981, LaiRobbins1982} and further generalized by Lai and Wei \cite{LaiWei1982}.
This analysis is a crucial step in the field of online estimation and optimization.
Lai \cite{Lai2003} gives a comprehensive survey of several related developments and discussions.
Its application to revenue management is given by den Boer and Zwart \cite{denBoerZwart2013}, in which bounds on cumulative regrets are analyzed as a measure of performance.
Here, regret is defined to be the difference between the expected revenue and the optimal revenue.
We develop these models for the setting of insurance where there are both demand and claims.
We, then, consider a second approach to this problem based on a Bayesian optimization.
It was proposed by Mockus \cite{Mockus1989, Mockus1994} for optimizing an unknown function using a Gaussian Process.
A Gaussian Process is a generalization of the Gaussian probability distribution, where random variables are modelled by stochastic processes.
Over the last two decades, GPs have been widely used in machine learning.
We investigate the upper-confidence bound (UCB) approach taken by Srinivas et al.\@ \cite{Srinivas2010, Srinivas2012}.
By maximizing the UCB aquisition function, we can determine the price at each time period.
For more details on Gaussian Process regression and Bayesian optimization in general, we refer to Rasmussen and Williams \cite{RasmussenWilliams2006} and Brochu \cite{Brochu2010}.

In a summary, the contributions of this work are as follows:
\begin{itemize}
    \item{
    We address the dynamic pricing problem with unknown demand \& claims by adaptive Generalized Linear Models and Gaussian Process regression approaches. To the best of our knowledge this paper is the first to consider online learning in the context of insurance pricing.}
    \item{
    In the GLM setting, based on den Boer and Zwart \cite{denBoerZwart2013}, we extend the pricing algorithm to insurance pricing by subtracting heavy-tail distributed claims.
    }
    \item{
    In the GP setting, we follow Srinivas et al.\@ \cite{Srinivas2010} for Bayesian optimization.
    GP with an alternative UCB function and additive kernel is applied to select the optimal price.
    }
    \item{
    We use cumulative regret to measure the performance of our algorithms, named GLM pricing algorithm and GP pricing algorithm.
    These have the following bounds:
    \begin{itemize}
        \item{
        The GLM pricing algorithm can achieve regret
        $O\left(\sqrt{T \log T}\right)$,
        }
        \item{
        The GP pricing algorithm has regret
        $ O \left(\sqrt{\gamma_{T} T \log T}\right)$.
        }
    \end{itemize}
    Here, $T$ is the length of the selling horizon and $\gamma_{T}$ is maximum information that the algorithm could learn about the demand and total claims functions.
    }
    \item{
    By our analysis, we show that these two mechanisms are simple, implementable and have good performance.
    }
\end{itemize}

Dynamic pricing and online learning have been successfully applied in a variety of industries such as airline ticketing, hotel bookings, car rentals, and fashion.
%As we discuss in the literature there are many studies on insurance pricing (with known demand and claims distribution);
However, to the best of our knowledge, online learning has not been applied to insurance pricing in any literature. %{we do not know of any literature applying online learning to insurance pricing}.
Thus, motivated by the powerful machine learning techniques and the growing applicability insurance industry,
this paper is among the first to investigate these methods to address problems in insurance pricing.
As insurance increasingly sold online and with insurance products continually changing, we believe that these methods will be important for actuaries now and in the future \cite{swissre2015}.

%There is an increased use of more bespoke insurance such as pay-as-you-drive or more personalized insurance for this reason it becomes more important to be able to jointly predict and optimize insurance products. Further the as discussed earlier, the increased use of online retail allows for much better utilization of these methods than was previously possible.

\subsection{Related Literature}
%Outline:
In this section, we provide a brief review of insurance pricing and dynamic pricing.
We also highlight related work on applying online learning to two statistical models: Generalized Linear Models and Gaussian Processes.
Finally, we discuss the previous work on revenue management with uncertainty.
%Background setting
\paragraph{Insurance Pricing and Dynamic Pricing}
%\Msout{Early insurance pricing can be traced back to the 1600s, when the first actual insurance company, known as ``The Insurance Office" was set up in London.}
Many researchers such as B{\"{u}}hlmann \cite{Buhlmann1970}, McClenahan \cite{McClenahan2001}, Jong and Heller \cite{JongHeller2008} point out that mathematical and statistical methods are needed to support actuaries to make pricing decisions.
%Many methods have been proposed in pricing non-life insurance products.
The linear models have been applied extensively in actuarial work.
For example, early literature uses linear models in motor insurance, see Baxter \cite{Baxter1980} and Coutts \cite{Coutts1984}.
In 1960, Bailey and Simon \cite{BaileySimon1960} introduce the minimum bias technique in classification ratemaking, which is an important milestone in non-life insurance pricing development \cite{David2015}.
In the 1980s, British actuaries introduced GLMs to insurance pricing and this has now become a standard approach in many countries \cite{OhlssonJohansen2010}.
A good overview of the use of GLMs in different situations in actuarial work is available in Haberman and Renshaw \cite{HabermanRenshaw1996}; for further research applying GLMs in insurance pricing see \cite{JongHeller2008,OhlssonJohansen2010,KaasGDD2009,Frees2010}.
In the last few years, the non-life insurance market has changed due to the increase of online services.
Machine learning techniques have become more popular in applications in the insurance sector.
These enhance and supplement the standard GLMs analysis.
We refer to W{\"{u}}thrich and Buser \cite{WuthrichBuser2018} for an overview and insight into GLMs and machine learning methods in non-life insurance pricing.

Dynamic pricing is the study of how demand responds to prices in a changing environment.
In recent decades, interest in dynamic pricing has grown rapidly.
Early profit optimization problems assume sellers have complete knowledge of the market, which means demand functions are known or can be found from previous selling experience.
%\Msout{Early work dates back to Cournot \cite{Cournot1838} in 1838.}
Evans \cite{Evans1924, Evans1930} is one of the first to propose a dynamic pricing model by adding time derivatives of prices to a static model.
Greenleaf \cite{Greenleaf1995}  numerically shows the significant effects of reference prices and develops an optimal dynamic pricing strategy in a monopoly setting.
Kopalle et al.\@ \cite{KopallePraveen1996} analytically generalize these results to a duopoly and an oligopoly settings.
Following this, Fibich et al.\@ \cite{FibichGadi2003}
then calculate explicitly the optimal pricing strategy in various nonsmooth optimization problems.
All of these works assume the demand function of consumers is deterministic and known.
Surveys by Aviv and Vulcano \cite{AvivVulcano2012} and den Boer \cite{DenBoer2015} provide an excellent  overview of this area.

% MQLE&LSE
\paragraph{Adaptive Generalized Linear Models}
Nelder and Wedderburn \cite{NelderWedderburn1972} first introduce Generalized Linear Models (GLM), which is an extension to classical linear regression.
As discussed above, it has become a well-established and standardised statistical technique to price the insurance products
\cite{OhlssonJohansen2010, Wuthrich2017}.
In the GLM framework, maximum likelihood estimation is a commonly used technique to find the parameters of a given Generalized Linear Models.
%\Mila{The maximum likelihood estimation requires complete information, which is not always possible \cite{Heyde1997}.}
%\Neil{Try and think what to put here}.
%To address this problem,
Wedderburn \cite{Wedderburn1974} proposes a method named quasi-likelihood estimation, an extension of likelihood estimations but only the first two moments of the observations are needed.
McCullagh and Nedler \cite{McCullaghNelder1989}
then apply GLMs with quasi-likelihood estimation to insurance ratemaking.
They fit a GLM to different types of data, including average claim costs for a motor insurance portfolio and claims frequency for marine insurance.
%\sout{Jong and Heller  \cite{JongHeller2008}, Kaas et al.\@ \cite{KaasGDD2009}, Frees \cite{Frees2010} and Ohlsson and Johansen \cite{OhlssonJohansen2010}}.

%explain to use GLM in non–life insurance pricing.

To price products, often a certainty equivalence rule is used. Here the optimal price is chosen for the estimated parameters. Thus when optimizing we treat estimates as if they were the true (unknown) parameters of the model.
Anderson and Taylor \cite{AndersonTaylor1976} apply a certainty equivalence rule to solve a multiperiod control problem.
However, strong consistency may not hold when applying a certainty equivalence rule to maximum quasi-likelihood estimates \cite{LaiRobbins1982, Lai2003}.
To deal with this problem, conditions are proposed to ensure the strong consistency for parameters estimators \cite{LaiRobbins1979, LaiRobbinsWei1979}.
Lai and Robbins \cite{LaiRobbins1981} introduce further conditions for an adaptive design, and
Lai and Wei \cite{LaiWei1982} generalize these conditions to multiple regression models with errors given by a martingale difference sequence.
%and prove that the minimum requirement is when the ratio of the minimum eigenvalue to the logarithm of the maximum eigenvalue goes to infinity.
%Chen et al.\@ \cite{ChenHY1999} apply the local inverse function theorem and obtain the strong convergence for maximum quasi-likelihood estimators of regression parameters in GLMs, under both fixed and adaptive designs .
Chen et al.\@ \cite{ChenHY1999} extend the results of \cite{LaiWei1982, LaiRobbinsWei1979} to GLMs under both fixed and adaptive designs.

%% MAB & BO %%
\paragraph{Multi-armed Bandits and Bayesian Optimization}

A multi-arm bandit problem refers to a broad class of sequential decision making problems.
At each time step, one must choose an arm amoungst a set of arms, each of which has unknown rewards.
There is a trade-off between exploration, i.e. estimating the distribution of rewards for all arms in the past, and exploitation, i.e. choosing the arm with higher expected reward.
Bubeck and Cesa-Bianchi \cite{BubeckBianchi2012} present a comprehensive review of work on multi-armed bandit problems.
In multi-armed bandits problem, the upper confidence bound (UCB) rule is commonly used to select arms at each time period.
The UCB algorithm constructs a confidence interval for the mean of each arm, and then chooses the arm that maximizes revenue under this estimation.
The UCB strategy is introduced by Auer et al.\@ \cite{Auer2002} to address a specific bandit model, and is used for asymptotic analysis of regret as first discussed in Lai and Robbins \cite{LaiRobbins1985}.
%\paragraph{Bayesian Regret Bound}
Regret bounds for multi-armed bandits problems have attracted a great deal of interest in different cases, such as linear models \cite{Dani2008, RusmevichientongTsitsiklis2010}, Generalized Linear Models \cite{Filippi2010}, Lipschitz functions \cite{Bubeck2011, Kleinberg2008}, Gaussian Process \cite{Srinivas2010} and Thompson Sampling \cite{Agrawal2012, Agrawal2012b, Russo2013}.
In the insurance context, each price is an “arm” and its revenue is the “reward”.

%\paragraph{Gaussian Process and Bayesian optimization}
Bayesian optimization \cite{Mockus1978} provides an efficient approach to address global optimization of an unknown potentially random or noisy function.
It is applicable and efficient when objective functions are unknown or are expensive to evaluate.
There are two significant stages in Bayesian optimization.
The first stage is to learn the objective function from available samples.
Bayesian optimization typically works by assuming the unknown function is sampled from a Gaussian Process (GP) \cite{Snoek2012}.
The second stage is to optimize a acquisition function to determine the next sampling points for the evaluation of the objective function.
High acquisition function values occur either because there is large uncertainty in the objective function (exploration) or a high prediction given by the model (exploitation).
%  Extension of GP
Srinivas et al.\@ \cite{Srinivas2010} consider a GP-based Bayesian optimization.
In this work, the authors propose a Gaussian Process upper confidence bound (GP) algorithm, where they sample the reward function from a GP and apply a UCB algorithm to bound the regret.
They achieve sublinear regret in terms of the maximum information gain, the maximum amount of informarion the algorithm could learn about the reward function.
For a comprehensive review of the Bayesian optimization and its applications, we refer to Brochu et al.\@ \cite{Brochu2010}.
To the best of our knowledge this is the first paper to consider Bayesian optimization in the context of insurance.

%% RM unknown demand %%
\paragraph{Revenue Management with Unknown Demand}
Finally we discuss developments on revenue management with uncertainties.
Gallego and van Ryzin \cite{GallegoRyzin1994} introduce a single-product dynamical pricing to revenue management.
Subsequent works have adapted this model to allow for unknown demand.
One popular case is to consider a parametric setting, where demand can be modeled with fixed but not known parameters.
Aviv and Pazgal  \cite{AvivPazgal2005} are among the first to consider model uncertainty.
They derive a closed form model with a single unknown parameter and assume that the arrival of consumers follows a Poisson distribution.
Harrison et al.\@ \cite{HarrisonJ2012} improve the learning and profit performance by a new method named the myopic Bayesian policy.
Broder and Rusmevichientong  \cite{BroderRusmevichientong2012} present a maximum-likelihood based model for the analysis of regret in dynamic pricing problems with a general parametric model.
They show that in a general case, upper bound of the $T$-period regret is $O(\sqrt{T})$.
den Boer and Zwart \cite{denBoerZwart2013} propose a controlled variance pricing policy, in which they create taboo intervals around the average of previously chosen prices to ensure sufficient price dispersion.
This policy is the first to consider the parametric model with unknown demand by maximizing the quasi-likelihood estimation.
They obtain an asymptotic upper bound on $T$-period regret as $O(T^{1/2+\delta})$, where $\delta >0$ is extremely small.
This work forms the base of our insurance pricing model.

% general  model
The pricing problem can also be addressed in a nonparametric way.
Kleinberg and Leighton  \cite{KleinbergLeighton2003}
provide an analysis of an online auction and introduce regret to measure of the performance of a pricing strategy.
Cope \cite{CopeEric2007} applies a nonparametric Bayesian approach using Dirichlet distributions as priors to achieve a revenue-maximizing goal in an e-commerce market.
Rusmevichientong et al.\@ \cite{RusmevichientongVanRoy2006} develop a nonparametric approach to a multiproduct pricing problem based on a real automobile data set.
Besbes and Zeevi \cite{BesbesZeevi2009, BesbesZeevi2012} use blind pricing policies to balance exploration-exploitation trade-offs and achieves asymptotically optimal.

\paragraph{Bandit Online Learning Problem} Traditionally, delays can be considered as a fixed constant.
Under this setting, Dudik et al.\@ \cite{Dudik2011} provide an efficient algorithm for stochastic contextual bandits and show that regret is additive.
Chapelle and Li \cite{ChapelleLi2012} present the influence of delayed feedback for contextual bandits in news article recommendation.
Cesa--Bianchi et al.\@ \cite{CesaBianchi2016} study networks of nonstochastic bandits.
Pike-Burke et al.\@ \cite{PikeBurke2018} discuss the case with delayed, aggregated anonymous feedback
and the expected delay is known.
They assume only the sum of regret is available while individual regret is unknown.
In general, delays may be a stochastic process. Agarwal and Duchi \cite{AgarwalDuchi2011} analyze stochastic gradient-based optimization algorithms  when delays are i.i.d randomly distributed.
Desautels et al.\@ \cite{Desautels2012} study parallel experiments with a bounded delay between an experiment and observation in a Gaussian Process bandit problems.
Vernade et al.\@ \cite{Vernade2017} consider infinite stochastic delays where some feedback can not be observed after a threshold.
For a systematic study of online learning with delayed feedback and the effects of delay on regret, we refer to Joulani et al.\@ \cite{Joulani2013}.
In their work, they show that delays additively increases regret in stochastic problems without requiring knowledge on distributions of delays.

\subsection{Organization}
%\paragraph{The remainder of the paper}
The sections of the paper are structured as follows.
In Section \ref{section_ModelsAssumptions}, we describe the optimization pricing problem in the insurance setting and define our pricing models.
We study GLM and GP models, as well as assumptions and estimation methods associated with each of these models.
In Section \ref{section_Pricingpolicy}, we propose GLM and GP pricing algorithms and explain how they work, respectively.
The main result of this paper is presented in Section \ref{section_RegretBounds}.
We consider bounds on cumulative regret, which help to measure the performance of each pricing policy.
In Section \ref{section_DelayedCase}, we extend both models with unknown delayed claims.
Section \ref{section_NumericalExamples} illustrates an experimental set-up and numerical results .
Finally, a conclusion and discussion of future work are provided in Section \ref{section_Conclusion}.
Auxiliary results and proofs are gathered in the \nameref{Appendix}.

%%%%%%%%%%%%%%%%%%%%%%%%%%%%%%%%%%%%%%%%%%%%%
%% General Setting: Pricing Models %%
%%%%%%%%%%%%%%%%%%%%%%%%%%%%%%%%%%%%%%%%%%%%%
\section{Models and Assumptions} \label{section_ModelsAssumptions}
In this section, we introduce two regression models and important assumptions.
We give a brief overview of the insurance pricing problem in Section \ref{OverviewProblem}.
In Section \ref{section_ParaModel}, we discuss an adaptive Generalized Linear Model (GLM), which is an parametric model, and explain how to estimate unknown parameters by quasi-likelihood estimation.
This model is built on ideas of den Boer and Zwart \cite{denBoerZwart2013}, and Lai and Wei \cite{LaiWei1982}.
In Section \ref{section_NonparaModel}, we introduce an adaptive Gaussian Process (GP) model with an UCB rule.

%% Overview Pricing Problem
\subsection{Overview}\label{OverviewProblem}

We consider an insurance company which sells a single product over a selling horizon $T > 0$.
The selling price $p_{t}$ is determined at the beginning of each time period $t \in \{0,\dots,T\}$.
We define the set of acceptable prices by $\mathcal{P} = [p_l, p_h]$, where $0 < p_l < p_h$ are the minimum and maximum selling prices.

%Define $D(p)$
We assume that dynamic pricing is only associated with past prices.
Given a determined selling price $p_{t}$ at time period $t$, the insurance company observes demand $d_{t} = D_{t} (p_{t})$, which is independent realisations of the random demand function $D(\cdot)$ for the selling price $p_{t}$.
%Define $C(p)$
Similarly, we denote the total claims as $C_{t} (p_{t})$ during the time period $t$ and observe the total claims $c_{t} = C_{t} (p_{t})$.
(Often we will suppress the sub-script $t$ from i.i.d. random variables $D(\cdot)$ and $C(\cdot)$.)

%define the expected revenue
In insurance, the premium is the expected income that the insurance company earns and claims are the amount that the insurance company loses.
If the selling price is known, the revenue collected in a single time period $t$ is $p_{t} d_{t} - c_{t}$.
The expected revenue at time $t$ is given by
\begin{align}\label{equ:revenuep}
    r(p_{t})
    = \Exp[p_{t} D(p_{t}) - C(p_{t})] \,.
\end{align}
Both demand and total claims respond to changes in prices at each time period simultaneously.
Once the price is specified, we assume the demand and total claims are independent of each other.
The insurance company aims to find an optimal pricing policy that generates maximum revenue, based on previous selling prices $\{p_i: \, i = 1, \dots, t-1\}$ and observations $\{d_i \,, c_i: \, i = 1, \dots, t-1\}$.
% Define cumulative regret
We use cumulative regret to measure the performance of pricing policies.
The regret is the expected revenue loss caused by not using the optimal price. More formally,
we define the cumulative regret over time horizon $T$ as
\begin{align}\label{eq:regretT}
    \regret(T)
    &=
    \Exp
    \left[
    \sum_{t=1}^{T}
    {r({p}^{*})-r({p}_{t})}
    \right]  \,.
\end{align}
Here, $r({p}^{*})$ is the revenue generated by the optimal price ${p}^{*}$:
\begin{equation}\label{eq:optimalprice}
    p^{*}
    =
    \mathop{\arg  \max}_{p \in \mathcal{P}} \,
    r(p_{t})\, .
\end{equation}
The objective of the seller is to maximize the sum of revenue, that is, to minimize the cumulative regret.

\subsection{Generalized Linear Pricing Model}\label{section_ParaModel}
We first consider the dynamic optimization pricing problem in a GLM  setting.
Here the expected revenue can not be calculated directly because it depends on unknown parameters that must be inferred.
We apply the maximum quasi-likelihood estimation (MQLE) to estimate the unknown parameters in the model.
We are concerned about the strong consistency for MQLE of regression parameters in the Generalized Linear Models (GLMs).

Our model is based on the work of den Boer and Zwart \cite{denBoerZwart2013}.
However, we consider a single insurance product with demand and heavy-tailed claims claims.
Here, we use the log of claims to describe large insurance claims.

\subsubsection{Model and Assumptions}
We assume that the insurance company knows the functional forms of the first two moments of demand and claims.
The model for demand distribution at time $t$ is given by
\begin{align*}
    \begin{split}
        \Exp[D(p_{t})]&=h_{1} \left(a_0+a_1 p_{t} \right) \, ,
        \\
        \Var(D(p_{t}))&=\sigma_{1}^{2} v_{1} (\Exp[D(p_{t})])
        \, .
    \end{split}
\end{align*}
Similarly, we assume the log of total claims is with expectation and variance, given by
\begin{align*}
    \begin{split}
        \Exp[\log C(p_{t})]&=h_{2}\left(b_0 + b_1 p_{t}\right) \, ,
        \\
        \Var(\log C(p_{t}))&=\sigma_{2}^2 v_{2} (\Exp[\log C(p_{t})])
        \, .
    \end{split}
\end{align*}
Here, parameters $a_0, a_1$ and $b_0, b_1$ are all unknown.
Notice here that we take the logarithm of total claims which is slightly non-standard when compared with results in revenue management.
In the context of insurance this can be used to model heavy-tailed claims distributions, such as the log-normal distribution.

We consider functions $h_{1}(\cdot), h_{2}(\cdot)$ are known link functions of price $p$ and unknown parameters.
The variance functions $v_{1}(\cdot), v_{2}(\cdot)$ are the variance of the expected demand \& claims.
The variances of the randomly distributed demand \& claims are functions of the variance functions with constants $\sigma_{1}, \sigma_{2} > 0$.
Functions $h(\cdot)$ and $v(\cdot)$ are twice continuously differentiable with first and second derivatives denoted by $\dot{h}(\cdot), \ddot{h}(\cdot)$ and $\dot{v}(\cdot), \ddot{v}(\cdot)$, respectively.
The link function is called a canonical link function when $\dot{h}(x) = v (h)$, otherwise it is called a general link function.

Denote $\bm{a} = (a_{0}, a_{1})^{\top}$ and $\bm{b} = (b_{0}, b_{1})^{\top}$, the expected revenue in \eqref{equ:revenuep} can be written as
\begin{align*}
    r(p_{t}) = r(p_{t}, \bm{a}, \bm{b}) \,.
\end{align*}
Moreover, the cumulative regret in \eqref{eq:regretT} after $T$ time periods becomes
\begin{equation*} \label{Regretab}
    \regret(T)
    =
    \Exp
    \left[
    \sum_{t=1}^{T}
    {r({p}^{*}, \bm{a}, \bm{b})-r(p_{t}, \bm{a}, \bm{b})}
    \right]
    \, .
\end{equation*}
Here, the optimal price ${p}^{*}$ is defined in \eqref{eq:optimalprice}.
%Now the objective of the seller is to find a pricing policy $\psi$ that minimizes $\regret(T)$, instead of maximizing the total expected revenue over a finite number of $T$ time periods.

Finally, we define the design matrix $P_{t}$ to be the sum of the transpose matrices achieved from price vectors
$\bm{p}(i) = (1, p_i)^{\top}$ for $i = 1, \dots, t$.
%Define a compact, convex, non-empty set $\mathcal{P} = \{1\} \times [p_l, p_h]$.
For $\bm{p} \in  \mathcal{P} \times \{1\}$,
the design matrix is given by
\begin{equation*}\label{designMatrix}
        P_{t}
        =
        \sum_{i=1}^{t}
        \bm{p}(i)\bm{p}(i)^\top \, .
\end{equation*}
We denote the largest eigenvalue of the design matrix $P_{t}$ as $\lambda_{\max} (t) = \lambda_{\max} (P_{t})$ and denote the smallest eigenvalue as $\lambda_{\min} (t) = \lambda_{\min} (P_{t})$.
%\Mila{\sout{We use $\lambda_{\min}(t)$ to measure the performance of pricing policy.}}

\subsubsection{Estimation of unknown parameters}\label{Estimation of unknown parameters}

The optimal policy cannot be calculated directly because regret depends on unknown parameters $\bm{a}, \bm{b}$.
To simplify the notations, we define parameter matrix as $\bm{\beta} = \left(\bm{a}, \bm{b}\right)$.
And we use $\bm{\beta}_{0}$ to denote the true values of regression parameter $\bm{\beta}$.
The maximum quasi-likelihood estimators, denoted by $\widehat{\bm{\beta}}_{t}$, are solutions to
\begin{equation}\label{eq:MQLE_beta}
    \bm{l}_{t}(\widehat{\bm{\beta}}_{t})
    =
    \sum_{i=1}^{t}
    \frac{\dot{h}(\bm{p}^\top(i)\widehat{\bm{\beta}}_{t}  )}{\sigma^2 v ( h(\bm{p}^\top(i)\widehat{\bm{\beta}}_{t}))}
    \bm{p}(i)
    \left(y_{i}-{h \left(\bm{p}^\top(i)\widehat{\bm{\beta}}_{t}\right)} \right)
    =\bm{0} \, .
\end{equation}
Let the filtration
$(\mathcal{F}_{t})_{t\in \N}$ be generated by $\{p_i \,,d_i \,, c_i: \, i = 1, \dots, t-1\}$ for each $t$.
Write
$\eta_{i} = y_{i} - {h (\bm{p}^\top(i)\widehat{\bm{\beta}}_{t})}$.
The error terms $\eta_{i}$ form a martingale difference sequence w.r.t.\@ $\mathcal{F}_{t}$, that is,
$\eta_{i}$ is $\mathcal{F}_{t}$-measurable and $\Exp[\eta_{i} \,| \,\mathcal{F}_{i-1}] = 0 $.
We also assume that for some $\gamma > 2$ almost surely,
\begin{enumerate}[({A}1)]
    \item{
    \begin{equation*}
        \sup_{i\in \mathbb N} \Exp[\eta_{i}^{2} \,| \,\mathcal{F}_{i-1}]
        \leq \sigma^2 < \infty \quad \text{and} \quad
        \sup_{i\in \mathbb N} \Exp[|\eta_{i}|^{\gamma}] < \infty  \,.
    \end{equation*}
    }\label{Assum_A1}
    \item{
    \begin{equation*}
        \lambda_{\min}(t) \to \infty
        \quad  \text{ and } \quad
        \log \lambda_{\max} (t) = o(\lambda_{\min} (t)) \,.
    \end{equation*}
    }\label{Assum_A2}
\end{enumerate}

%%%%%%%%%%%%%%%%%%%%%%%%%%%%%%%%%%%%%%%%%%%%%%
%% GP Pricing %%
%%%%%%%%%%%%%%%%%%%%%%%%%%%%%%%%%%%%%%%%%%%%%%
\subsection{Gaussian Process Pricing Model}\label{section_NonparaModel}
Now, we construct a  Bayesian model by sampling the expected demand and expected total claims from Gaussian Processes (GP).
Our pricing model is an extension of Srinivas et al.'s\@ \cite{Srinivas2010} with an alternative UCB rule to the field of insurance where demands and claims are considered.

First, we offer a brief introduction of Gaussian Process regression and more complete details can be found in Rasmussen and Williams \cite{RasmussenWilliams2006}.
A Gaussian Process is a collection of random variables, any finite number of which have a joint Gaussian distribution.
It is completely specified by its mean function
$\mu({p})$
and covariance function (or kernel)
$k({p}, {p}')$ given by
\begin{align*}
    \begin{split}
        \mu({p}) & = \Exp [f({p})]\,, \\
        k({p}, {p}') & = \Exp [\left(f({p})-\mu({p})\right)\left(f({p}')-\mu({p}')\right)] \,.
    \end{split}
\end{align*}
Then, we can generate a Gaussian Process as
\begin{equation*}
    f({p}) \sim \mathcal{GP}(\mu({p}),  k({p}, {p}')) \,.
\end{equation*}
Without loss of generality, we assume that mean is a constant and covariance function is strictly bounded.

For a noisy sample
$\bm{y}_{T} = [y_{1}, \dots, y_{T}]^{\top}$, given a collection of input points $\{{p}_{1}, \dots, {p}_{T}\}$.
We define ${p}_{t}$ as the $t$-th sample and ${y}_{t} = f({p}_{t}) + \varepsilon_{t}$, here $\varepsilon_{t} \sim \mathcal{N}(0,\,\sigma^{2})$ is independent and identically distributed Gaussian noise with variance $\sigma^{2}$.
As Gaussian Process can describe a distribution over functions, we use $\mathcal{GP}(\mu({p}), k({p}, {p}'))$ as the prior distribution over $f$.
The posterior over $f$ is also a GP distribution with mean $\mu_{T}({p})$ and covariance function $k_{T}({p},{p}')$ given by
\begin{align}\label{posteriorGP}
    \begin{split}
        \mu_{T}({p})
        & = \bm{k}_{T}({p})^{\top} \left(\bm{K}_{T} + \sigma^{2}\bm{I}\right)^{-1} \bm{y}_{T} \,,\\
        k_{T}({p},{p}')
        & = k({p}, {p}') - \bm{k}_{T}({p})^{\top}\left(\bm{K}_{T} + \sigma^{2}\bm{I}\right)^{-1}\bm{k}_{T}({p}')\,,
    \end{split}
\end{align}
where $\bm{k}_{T}({p}) = [k({p}_{1},{p}), \dots, k({p}_{T},{p})]^{\top}$ and covariance matrix $\bm{K}_{T}$ is the positive
definite matrix whose entries are $\bm{K}_{i,j} = k({p_{i}},{p}_{j})$ for $i, j = 1, \dots, T$.

%several commonly-used covariance functions.
The kernel determines how observations influence the prediction of nearby or similarity inputs.
There are two commonly used kernels: the squared exponential kernel $k_{\sigma, l}$ and the Mat\'{e}rn kernel $k_{\nu, l}$, given by
\begin{align*}
    \begin{split}
        k_{\sigma, l}({p}, {p}')
        & = \exp \left( -\frac{1}{2l^{2}}|{p} - {p}'|^{2}\right) \,,\\
        k_{\nu, l}({p}, {p}')
        & =\frac{1}{2^{\nu-1}\Gamma(\nu)}\left( \frac{2\sqrt{\nu}|{p} - {p}'|}{l}\right)^{\nu} B_{\nu}\left( \frac{2\sqrt{\nu}|{p} - {p}'|}{l}\right) \,.
    \end{split}
\end{align*}
Here, $l$ is the length-scale and $\nu$ is the smoothness parameter.
Moreover, $\Gamma(\cdot)$ is the Gamma function and $B_{\nu}(\cdot)$ is the modified Bessel function of the second kind of order $\nu$.
Note that the Mat\'ern kernel reduces to the exponential kernel $k_{\nu, l}(r)= \exp \left(-{r }{l}\right)$ when the smoothness parameter $\nu = 1/2$ and reduces to the squared exponential kernel when $\nu \rightarrow \infty$.

%\subsubsection{Model and Assumptions}
%We sample the expected demand and expected total claims from Gaussian Processes.
We define the function of expected demand at price $p_{t}$ as $f_{d}(p_{t})$,
and similarly define the expected claims as $f_{c}(p_{t})$.
Functions $f_{d}(\cdot), f_{c}(\cdot)$ are independently sampled from GPs with known means $\mu_{d}, \mu_{c}$ and kernels $k_{d}(p,p'), k_{c}(p,p')$.
That is, $f_{d} \sim \mathcal{GP} (\mu_{d}, k_{d})$ and $f_{c} \sim \mathcal{GP} (\mu_{c}, k_{c})$.
The posteriors over $f_{d}$ and $f_{c}$ are GPs and also follow GP posterior update in \eqref{posteriorGP}.

The expected revenue function given a determined price $p_{t}$ at time $t$ is
\begin{align*}
    r(p_{t})
    = p_{t}\cdot f_{d}(p_{t}) - f_{c}(p_{t}) + \varepsilon_{t}^{r}\,.
\end{align*}
Here, the noise term $\varepsilon_{r} \sim \mathcal{N}(0, \sigma_{r}^2)$ is a combination of demand noise and claims noise.

We can see that $r(\cdot)$ is sampled from a GP as well, with an additive kernel.
This is because the sum of GPs is also a GP, and the kernel has the form of a direct sum.
Then, we have
$r \sim \mathcal{GP}(\mu_{r}, k_{r})$
with known $\mu_{r} = p \cdot \mu_{d} - \mu_{c}$ and $k_{r} = p^2 \cdot k_{d} + k_{c}$.
The cumulative regret over time horizon $T$ becomes
\begin{align*}
    \begin{split}
        \regret(T)
        & =
        \Exp
        \left[
        \sum_{t=1}^{T}
        \left({p}^{*} \cdot f_{d}({p}^{*}) - f_{c}({p}^{*})\right)
        -\left({p}_{t} \cdot f_{d}({p}_{t}) - f_{c}({p}_{t})\right)
        \right]  \,.
    \end{split}
\end{align*}
Our model is an extension of Srinivas et al.'s\@ \cite{Srinivas2010} work to the field of insurance.
Srinivas et al.\@ \cite{Srinivas2010} study a only one function $f$.
In our case, we consider an additive form, that is the revenue function $r$ contains two components: $f_{d}(\cdot)$ and $f_{c}(\cdot)$ and is sampled with an additive kernel.
Here, the samplings of $f_{d}(\cdot)$ and $f_{c}(\cdot)$ are independent of each other.
%%%%%%%%%%%%%%%%%%%%%%%%%%%%%%%%%%%%%%%%%%%%%
%% General Setting: Pricing Policy %%
%%%%%%%%%%%%%%%%%%%%%%%%%%%%%%%%%%%%%%%%%%%%%
\section{Pricing Policy} \label{section_Pricingpolicy}
In this section, we propose two algorithms called ``GLM Pricing Algorithm" and ``GP Pricing Algorithm" in Section \ref{section_ParaAlgo} and Section \ref{section_NonparaAlgo}, respectively.

A popular pricing policy is certainty equivalent pricing, proposed by Anderson and Taylor in a simple linear regression model \cite{AndersonTaylor1976} and further developed by the authors in a more general multiple regression model \cite{AndersonTaylor1979}.
We denote the certainty equivalent price by $\bm{p}_{cep}$.
It is efficient provided current parameter estimates are correct.
Given MQLE $\widehat{\bm{\beta}}_{t}$,  $\bm{p}_{cep}$ is the price that maximizes the expected revenue, given by
\begin{equation}\label{eq:CEP}
    \bm{p}_{cep}
    =
    \mathop{\arg  \max}_{ \bm{p} \in \mathcal{P} } \ r(\bm{p}, \widehat{\bm{\beta}}_{t})\, .
\end{equation}
Certainty equivalent pricing is popular, essentially because it separates the statistical problem from the problem of optimizing revenue and reward.
However, it is possible that parameter estimates converge to incorrect values due to convergence being too quick.
Specifically, Lai and Robbin \cite[Section 2]{LaiRobbins1982} prove that inconsistency may occur for a linear demand function with constant variance under an iterative least squared policy.
 den Boer and Zwart \cite{denBoerZwart2013} further relax the assumptions on the values of prices and propose a variant certainty equivalent pricing: controlled variance pricing policy.  By creating taboo intervals around mean prices that have been chosen, they choose the next price outside these intervals to obtain more information.
Keskin and Zeevi \cite{KeskinZeevi2014} also point out that the certainty equivalent pricing has poor performance of due to the existence of “uninformative” price and therefore introduce a constrained iterative least squared policy.

%  Pricing
\subsection{Adaptive GLM Pricing}
\label{section_ParaAlgo}
Based on the work of den Boer and Zwart \cite{denBoerZwart2013},
we propose a dynamic pricing policy with an additional constraint on the smallest eigenvalue of the design matrix with the aim to control the convergence of the pricing policy.
If we ensure a lower bound on $\lambda_{\min}(t)$, we can say that there is sufficient price dispersion guaranteeing the strong convergence of the MQLE.
However, there is no simple explicit relation between
$\lambda_{\min}(t)$ and $\lambda_{\min}(t+1)$.
We introduce the inverse of the trace of the inverse design matrix $\tr(P_{t}^{-1})^{-1}$.
For any $n \times n$ positive definite matrix $A$, we have
$\tr(A^{-1})^{-1} \leq \lambda_{\min}(A) \leq n \tr(A^{-1})^{-1}$.

Let $\bm{v}_{1}, \bm{v}_{2}$ be associated unit eigenvectors, which are an orthonormal basis of $\mathbb{R}^{2}$.
For $t > 2$, the optimal price can be written as a linear combination of these unit eigenvectors, that is $\bm{p}_{cep} = \sum_{i=1}^{2}\alpha_{i}\bm{v}_{i}$.

Let $\mathcal{L}$ be a class of positive differentiable monotone increasing functions $L$ such that $L(t)\to \infty$ and
$t \to \frac{1}{L(t)}$ is convex.
Choose a function $L_1(t) \in \mathcal{L}$, and let $\tr(P_{t}^{-1})^{-1}\geq L_{1}(t)$ for all $t \geq 2$,
then we have
\begin{equation*}
    \lambda_{\min}(t) \geq L_{1}(t) \,.
\end{equation*}
A proper choice of function
$L_1(t)$ guarantees sufficient price dispersion, and therefore, guarantees
the convergence of parameter estimates to the true parameters, as well as the asymptotical convergence of our price sequence to the optimal price.
Specifically, we present the optimal regret bounds when $L_{1}(t) = c \sqrt{t\log t}$ for some $c >0$.

Now, we show how the pricing algorithm works.
First, we choose three linearly independent initial price vectors to estimate the unknown parameters ${\widehat{\bm{\beta}}_{t}}$ by \eqref{eq:MQLE_beta}.
If the MQLE ${\widehat{\bm{\beta}}_{t}}$ does not exist or the constraint on the price dispersion $L_{1}(t)$ is not met, we repeat previous prices until we can find solutions of ${\widehat{\bm{\beta}}_{t}}$ to the MQLE and satisfy the sufficient price dispersion requirement as shown in \eqref{betanotexist} in Algorithm~\ref{Algo_GLMPricing}.
Condition $\tr ({P_{t+j}}^{-1})^{-1} \geq L_{1}(t+j)$ will eventually be met because $j$ is always finite.
If the MQLE ${\widehat{\bm{\beta}}_{t}}$ exist and the constraint on the price dispersion $L_{1}(t)$ is also satisfied, we set the next price to be the certainty equivalent price $\bm{p}(t+1) = \bm{p}_{cep}$, and check whether the condition in \eqref{eq:traceL1} in Algorithm~\ref{Algo_GLMPricing} holds or not.
If it does not hold, we then choose $\bm{p}(t+1) =  \bm{p}_{cep} + \bm{\phi}_{t}$.
Here,
\begin{equation}\label{eq:phi}
    \bm{\phi}_{t} = K \sqrt{\dot{L}_{1}(t)} \left(v_{2,1}\bm{p}_{cep} - \bm{v}_{2} \right) \,,
\end{equation}
and $v_{2,1}$ is the first component of $\bm{v}_{2}$.
Constant $K>0$ must satisfy $\left|K \sqrt{\dot{L}_{1}(t)}\right| \leq 1 $.
These pricing steps are given in Algorithm~\ref{Algo_GLMPricing}.

% Algorithm1
\begin{algorithm}[ht]
\caption{GLM Pricing Algorithm}

\SetKwInput{Initialisation}{Initialisation}
\Initialisation{}
{Choose $L_{1} \in \mathcal{L}$.}

{Choose linearly independent initial price vectors $\bm{p}(1), \bm{p}(2), \bm{p}(3)$ in $\mathcal{P}$.}

For all $t \geq 3$:

\SetKwInput{Estimation}{Estimation}
\Estimation{}
Calculate ${\widehat{\bm{\beta}}_{t}}$ using \eqref{eq:MQLE_beta}.

\SetKwInput{Pricing}{Pricing}
\Pricing{}
\begin{enumerate}[(I)]
    \item\label{betanotexist}{
    If ${\widehat{\bm{\beta}}_{t}}$ does not exist or $\tr(P_{t}^{-1})^{-1}\ngeq L_{1}(t)$,
    then set $\bm{p}(t+1) = \bm{p}(1), \cdots, \bm{p}(t+j) = \bm{p}(j)$,
    here $j$ is the smallest integer satisfies $\tr ({P_{t+j}}^{-1})^{-1} \geq L_{1}(t+j)$.
    }
    \item\label{eq:traceL1}{
    If ${\widehat{\bm{\beta}}_{t}}$ exists and $\tr(P_{t}^{-1})^{-1}\geq L_{1}(t)$, then we set $\bm{p}(t+1) = \bm{p}_{cep}$ and consider
    \begin{equation*}
        \tr\left(\left({P_{t}+ \bm{p}(t+1) \bm{p}(t+1)^{\top}}\right)^{-1}\right)^{-1} \geq  L_{1}(t+1)\,.
    \end{equation*}
    If it does not hold, we instead set $\bm{p}(t+1) = \bm{p}_{cep} + \bm{\phi}_{t}$ and $\bm{\phi}_{t}$ is defined in \eqref{eq:phi}.
    Here we can choose $\|\bm{\phi}_{t}\|^2 =  \dot{L}_{1}(t) \left(1+\max_{\bm{p}\in \mathcal{P}}\|\bm{p}\|^{2}\right)$, such that the above requirement is satisfied}.
\end{enumerate}
\label{Algo_GLMPricing}
\end{algorithm}

The following proposition guarantees that when prices are chosen, the price dispersion condition \eqref{eq:TraceL1t} is satisfied.
\begin{proposition}\label{proposition_PropGLMPricingNext}
If ${\widehat{\bm{\beta}}_{t}}$ exists and $\tr(P_{t}^{-1})^{-1}\geq L_{1}(t)$,
we set the next price to be $\bm{p}(t+1) =  \bm{p}_{cep} + \bm{\phi}_{t}$, then
\begin{equation}\label{eq:TraceL1t}
    \tr\left(\left({P_{t}+ \bm{p}(t+1) \bm{p}(t+1)^{\top}}\right)^{-1}\right)^{-1}
    \geq  L_{1}(t+1) \, .
\end{equation}

\end{proposition}
\begin{proof}
See \nameref{Appendix} \ref{Proof_PropGLMPricingNext}.
\end{proof}

%%  Pricing
\subsection{Adaptive GP Pricing Model}\label{section_NonparaAlgo}
% If using GP algorithm in \cite{Srinivas2010},
In the GP setting, we determine our pricing policy by the upper confidence bound (UCB) rule.
We start by reviewing Srinivas et al.'s work \cite{Srinivas2010}, where
the posterior GP is used to construct a UCB function.
At each time step $t-1$, they set the next sampling point to be the one that maximizes the UCB function, given by
\begin{equation*}
    {p}_{t} = \argmax_{{p} \in \mathcal{P}}{{\mu}_{t-1}({p})+ \sqrt{\varphi_{t}} {\sigma}_{t-1}({p})} \,.
\end{equation*}
Here, the posterior mean ${\mu}_{t-1}(\cdot)$ is obtained after $t-1$ observations and is the current estimate of $f$.
The posterior standard deviation ${\sigma}_{t-1}(\cdot)$ is the uncertainty associated with this estimate.
Term ${\mu}_{t-1}(\cdot)$ is the explicit exploitation on what we have known and ${\sigma}_{t-1}(\cdot)$ is the exploration on what we haven't known.
Parameter $\varphi_{t}$ is used to balance the trade-off between exploitation ${\mu}_{t-1}(\cdot)$ and exploration ${\sigma}_{t-1}(\cdot)$.

In our work, we update the UCB function with additive mean and kernels, defined by
\begin{equation*}
    {p}_{t} = \argmax_{{p} \in \mathcal{P}}{{\mu}_{t-1}^{r}({p})+ \sqrt{\varphi_{t}} {\sigma}_{t-1}^{r}({p})} \,.
\end{equation*}
Here,
\begin{align}\label{eq:rmusdv}
    \begin{split}
        {\mu}_{t-1}^{r} = p_{t-1}{\mu}_{t-1}^{d} + {\mu}_{t-1}^{c}  \,\quad
        {\sigma}_{t-1}^{r}({p})  = p_{t-1}{\sigma}_{t-1}^{d}({p}) + {\sigma}_{t-1}^{c}({p}) \,.
    \end{split}
\end{align}

Now, we present the implementation of GP algorithm for pricing.
At time $t-1$, the algorithm generates a price that maximizes the UCB function, which is the optimal price ${p}^{*}$.
Each price is determined by the history $\mathcal{F}_{t-1}$ and the policy followed by the UCB rule.
In the continuous price set, the optimal price ${p}^{*}$ exists and is unique.
Then, we set the optimal price to be the next sampling price ${p}_{t}$.
Given the price ${p}_{t}$,
we sample $f_{d}({p}_{t})$ and $f_{c}({p}_{t})$ individually.
Since the revenue is determined by two parts: premium ${p}_{t} \cdot f_{d}({p}_{t})$ and total claims $f_{c}({p}_{t})$ and chosen price can be considered as a constant, the sampled revenue is a linear combination of sampled functions $f_{d}({p}_{t})$ and $f_{c}({p}_{t})$.
Therefore, we can obtain the sampled revenue function.
Next, we perform the GP posterior update  \eqref{posteriorGP} to obtain posterior distributions for demand and total claims and update the posterior distribution of $f_{d}$ and $f_{c}$ given $\mathcal{F}_{t} = \mathcal{F}_{t-1}\cup \{p_{t}, d_{t}, c_{t}\}$, respectively.
Finally, we obtain the mean ${\mu}_{t}^{r}$ and standard variance ${\sigma}_{t}^{r}$ defined in \eqref{eq:rmusdv} and then updated the UCB function to determine the next selling price.
The pseudocode for pricing a new released insurance product via GP pricing algorithm is shown in Algorithm~\ref{algo_GPPricing}.

% Algorithm2
\begin{algorithm}[ht]
\caption{GP Pricing Algorithm}
\label{algo_GPPricing}
\KwIn{GP Prior ${\mu}_{0}, {\sigma}_{0}, k $}

\For{$t = 1, 2, \dots$}{
    Select price:
    ${p}_{t} \gets \argmax_{{p} \in \mathcal{P}}{{\mu}_{t-1}^{r}({p})+ \sqrt{\varphi_{t}} {\sigma}_{t-1}^{r}({p})}$ \;
    Sample the revenue function:
    ${r}_{t} \gets {p}_{t} \cdot f_{d}({p}_{t}) - f_{c}({p}_{t}) + \varepsilon_{t}^{r}$ \;
    Update estimate:
    \vspace{-0.5em}
    \renewcommand\labelitemi{$\vcenter{\hbox{\tiny$\bullet$}}$}
    \begin{itemize}
        \item{
        ${\mu}_{t}^{d}$ and ${\sigma}_{t}^{d}$  by performing the GP posterior update  \eqref{posteriorGP}\,\;
        }
        \vspace{-1em}
        \item{
        ${\mu}_{t}^{c}$ and ${\sigma}_{t}^{c}$ by performing the GP posterior update  \eqref{posteriorGP}\,;
        }
        \vspace{-1em}
        \item{
        Obtain $\mu_{t}^{r}$,  ${\sigma}_{t}^{r}$\,;
        }\vspace{-1em}
    \end{itemize}
    }
\end{algorithm}

%%%%%%%%%%%%%%%%%%%%%%%%%%%%%%%%%%%%%%%%%
%% Regret Bounds
%%%%%%%%%%%%%%%%%%%%%%%%%%%%%%%%%%%%%%%%%
\section{Bounds on the Regret}\label{section_RegretBounds}
In this section, we show theoretical analysis obtained by proposed approaches.
Our main results on cumulative regret are stated in Theorem \ref{theorem_GLMBound} in Section \ref{section_ParaResults} and Theorem \ref{theorem_GPBounds} in Section \ref{section_NonparaResults}, as well as proofs .

\subsection{Adaptive GLM Pricing}\label{section_ParaResults}
We show the strong consistency and convergence of the maximum quasi-likelihood estimators in Generalized Linear Models (GLMs).
Lai and Wei \cite{LaiWei1982} study the least-squares estimate in linear stochastic regression models.
Chen et al.\@ \cite{ChenHY1999} extend the results of Lai and Wei \cite{LaiWei1982} to  GLMs with canonical link functions.
Regret bounds depend on the lower bound on the smallest eigenvalue $\lambda_{\min}(t)$ of the design matrix $P_{t}$ and the expected value of difference between the chosen prices and optimal prices.

We will first show that studying the bound on the regret is equivalent to studying the bound on $\left\|{\bm{\beta}} -{\bm{\beta}}_{0}\right\|^2$.

\begin{proposition}\label{proposition_GLMBeta}
%\Neil{Why would there not be an open region arround $\bm \beta_0$ there are currently no constraints on the region that $\bm \beta$ belongs to.}

Assume there is an open, bounded neighbourhood
$V \in \R^{2 \times 3}$ around true $\bm{\beta}_{0}$, such that, for all $\bm{\beta} \in V$, we can find a unique optimal price $\bm{p}^{*}$ that maximizes the revenue function $r(\bm{p}, \bm{\beta})$.
Given $\bm{p}({\bm{\beta}_{0}}) \in \mathcal{P}$,
$\frac{\partial r(\bm{p}^{*}, \bm{\beta}_{0})}{\partial p} = 0$ and $\frac{\partial^2 r(\bm{p}, \bm{\beta}_{0})}{ \partial p^2} < 0$,
we can derive
\begin{equation}\label{eq:r&p}
    \left|r(\bm{p}, \bm{\beta}_{0})-r(\bm{p}^{*}, \bm{\beta}_{0})\right|
    =
    O\left( \left\|\bm{p}-\bm{p}^{*}\right\|^2 \right) \, .
\end{equation}
Further, if we assume there exists $t_0 \in \mathbb N$ such that $\widehat{\bm{\beta}}_t \in V$  for all $t \geq t_0$, then
\begin{equation}\label{eq:p&b}
    \left\|\bm{p}({\widehat{\bm{\beta}}_{t}})-\bm{p}(\bm{\beta}_{0})\right\|^2
    =
    O\left(\left\| \widehat{\bm{\beta}}_{t} - \bm{\beta}_{0}\right\|^2 \right) \, .
\end{equation}
\end{proposition}
\begin{proof}
See \nameref{Appendix} \ref{Proof_PropGLMBeta}\,.
\end{proof}
We focus on a simple case, where link functions are canonical, that is,
$\dot{h}(\cdot) = v (h(\cdot))$.
This gives
\begin{equation*}\label{eq:MQLE__Clink_beta}
    \bm{l}_{t}(\widehat{\bm{\beta}}_{t})
    =
    \sum_{i=1}^{t}
    \frac{1}{\sigma^2}
    \bm{p}(i)
    \left(y_{i}-{h \left(\bm{p}^\top(i)\widehat{\bm{\beta}}_{t}\right)} \right)
    =\bm{0} \, .
\end{equation*}
Under the Assumptions A\ref{Assum_A1} and A\ref{Assum_A2},
we show the MQLE $\widehat{\bm{\beta}}_{t}$ eventually exists and this estimator is also strongly consistent.

\begin{proposition}\label{proposition_GLMboundsClinkF}
Suppose Assumptions A\ref{Assum_A1} and A\ref{Assum_A2} are satisfied and
\begin{equation*}
    \sup_{\bm{p}\in{\mathcal{P}}}\|\bm{p}\|^{2} \leq r^{2} < \infty \,,
\end{equation*}
where $r = p_{h} - p_{l}$.
Then, MQLE $\widehat{{\bm{\beta}}}_{t}$ in (\ref{eq:MQLE_beta}) exists and
$\lim_{t \to \infty}{\widehat{{\bm{\beta}}}_{t}} = {{\bm{\beta}}}_{0}$ a.s. and
\begin{align*}
    \Exp\left[\left\|\widehat{{\bm{\beta}}}_{t}-{\bm{\beta}}_{0}\right\|^2\right]
    = O \left(\frac{\log(t)}{L_{1}(t)}\right) \,,
\end{align*}
given $\lambda_{\min}(t) \geq L_1(t)$.
\end{proposition}
\begin{proof}
See \nameref{Appendix} \ref{Proof_PropGLMboundsClinkF}.
\end{proof}
{Theorem \ref{theorem_GLMBound} provides an upper bound on the regret in terms of the function $L_{1}(t)$.}

\begin{theorem}\label{theorem_GLMBound}
Suppose there exists $t_{0} \in \N$, if MQLE $\widehat{{\bm{\beta}}}_{t}$ is  strongly consistent and following conditions are satisfied:
\begin{enumerate}
    \item{
    $\lambda_{\min}(t) \geq L_{1}(t) = c \sqrt{t \log(t)}$ a.s. for all $t \geq t_{0}$ and $c>0$,
    }
    \item \label{Ubd_2} {
    $\sum_{t=1}^{T} \left\| \bm{p}(t) - \bm{p}(\widehat{{\bm{\beta}}}_{t}) \right\|^2 \leq K L_{1}(T)$ a.s.\@ for all $T \geq t_{0}$ and some $K_1 >0$,
    }
\end{enumerate}
then the regret after $T$ time periods is
\begin{align*}
    \regret\,(T)
    & =
    O \left(L_{1}(T)+ \sum_{t=1}^{T} \frac{\log(t)}{L_{1}(t)}  \right) \, .
\end{align*}
\end{theorem}

\begin{proof}
By \eqref{eq:r&p} in Proposition \ref{proposition_GLMBeta},
the regret over selling horizon $T$ can be derived in terms of $\left\|\bm{p}(t)-\bm{p}^{*} \right\|^{2}$, that is
\begin{align*}
    \regret(T)
    & =
    O \left(
    \Exp\left[\sum_{t=t_{0}}^{T} \left\|\bm{p}(t)-\bm{p}^{*} \right\|^{2}\right]
    \right) \,.
\end{align*}
We can expand the expectation above as following and we will explain it step by step afterwards.
\begin{align*}
   \begin{split}
    \Exp\left[\sum_{t=t_{0}}^{T} \left\|\bm{p}(t)-\bm{p}^{*} \right\|^{2}\right]
    & \leq 2 \Exp\left[\sum_{t=t_{0}}^{T} \left\|\bm{p}(t)-\bm{p}(\widehat{\bm{\beta}}_{t}) \right\|^{2} \right]
    +
    2 \Exp\left[\sum_{t=t_{0}}^{T} \left\|\bm{p}(\widehat{\bm{\beta}}_{t}) - \bm{p}^{*}\right\|^{2} \right] \\
    & \leq 2K L_{1}(T)
    +
    2 K_{2} \Exp\left[\sum_{t=t_{0}}^{T} \left\|\widehat{\bm{\beta}}_{t} -  {\bm{\beta}}_{0}\right\|^{2} \right] \\
    & = O \left(L_{1}(T)+ \sum_{t=1}^{T}
    \left(
    \frac{\log(t)}{L_{1}(t)}
    \right)
    \right) \,,
   \end{split}
\end{align*}
here $K_{2}> 0$ is some non-random constant .
In the first inequality, we use
$(a+b)^2 \leq 2a^2 +2b^2$, that is
\begin{align*}
    \begin{split}
        \left\|\bm{p}(t)-\bm{p}^{*} \right\|^{2}
        & = \left\|\bm{p}(t)-\bm{p}(\widehat{\bm{\beta}}_{t}) + \bm{p}(\widehat{\bm{\beta}}_{t}) - \bm{p}^{*} \right\|^{2} \\
        & \leq
        2  \left\|\bm{p}(t)-\bm{p}(\widehat{\bm{\beta}}_{t}) \right\|^{2} +
        2 \left\|\bm{p}(\widehat{\bm{\beta}}_{t}) - \bm{p}^{*}\right\|^{2} \,.
    \end{split}
\end{align*}
In the second inequality, by Assumption \ref{Ubd_2} in Theorem \ref{theorem_GLMBound}, we bound the terms to obtain the first term
$$
2 \Exp\left[\sum_{t=t_{0}}^{T} \left\|\bm{p}(t)-\bm{p}(\widehat{\bm{\beta}}_{t}) \right\|^{2} \right]
\leq
2 K L_{1}(T) \,.
$$
By \eqref{eq:p&b} from Proposition \ref{proposition_GLMBeta}, we have
\begin{equation*}
    \left\|\bm{p}({\bm{\beta}}_{0})-\bm{p}(\widehat{\bm{\beta}}_{t}) \right\|^{2}
    \leq K_{2} \left\|  {\bm{\beta}}_{0}-\widehat{\bm{\beta}}_{t} \right\|^{2}  \, .
\end{equation*}
This gives the second term
\begin{equation*}
    2 \Exp\left[\sum_{t=t_{0}}^{T} \left\|\bm{p}(\widehat{\bm{\beta}}_{t}) - \bm{p}^{*}\right\|^{2} \right]
    \leq
    2 K_{2} \Exp\left[\sum_{t=t_{0}}^{T} \left\|\widehat{\bm{\beta}}_{t} -  {\bm{\beta}}_{0}\right\|^{2} \right] \,.
\end{equation*}
Finally, by Proposition \ref{proposition_GLMboundsClinkF} we can bound the regret by
\begin{equation*}
    \regret \, =
    O \left(L_{1}(T)+ \sum_{t=1}^{T}\frac{\log(t)}{L_{1}(t)}
    \right) \,.
\end{equation*}
\end{proof}
When $L_{1}(t) = c \sqrt{t \log(t)}$, we find an optimal pricing strategy that achieves $\regret (T) = O \left(\sqrt{T\log{T}}\right)$.
The regret bound that we obtained corresponds to the results in Kleinberg and Leighton \cite[Theorem 1.2]{KleinbergLeighton2003} but with different algorithms.

\subsection{Adaptive GP Pricing Model}\label{section_NonparaResults}
% Regret \regret ({T}) bounds
We now establish cumulative regret for the Gaussian Process (GP) Bayesian optimization, which is bounded by the maximum information gain $\gamma_{T}$.
Suppose the subset $\mathcal{A} = \{p_{1}, \dots, p_{T}\} \subset \mathcal{P}$ is finite.
Let $\bm{y}_{\mathcal{A}} = f(\bm{p}_{\mathcal{A}}) + \bm{\varepsilon}_{\mathcal{A}}$ denote the observations and let $f_{\mathcal{A}}$ denote the function values at these points, that is $f_{\mathcal{A}} = f(\bm{p}_{\mathcal{A}})$.
We introduce the Shannon Mutual Information and denote it as $I(\cdot)$.
For a Gaussian Process, $I\left(\bm{y}_{\mathcal{A}}; f_{\mathcal{A}}\right) = H(\bm{y}_{\mathcal{A}}) - H(\bm{y}_{\mathcal{A}}|f)$, where $H$ is the entropy.
In the multivariate Gaussian case,
$H(\mathcal{N}(\mu, \bm{\Sigma})) = \frac{1}{2} \log |2\pi e \bm{\Sigma}|$, so that
$I\left(\bm{y}_{\mathcal{A}}; f_{\mathcal{A}}\right) = \frac{1}{2} \log |\bm{I} + \sigma^{-2} \bm{K}_{\mathcal{A}}|$, where
$\bm{K}_{\mathcal{A}} = [k({p}, {p}')]_{\bm{p},\bm{p}'\in \mathcal{A}}$.
After $T$ rounds, the maximum information gain between $\bm{y}_{\mathcal{A}}$ and $f_{\mathcal{A}}$ is
\begin{equation}\label{eq:gamma}
    \gamma_{T}(k) = \max_{\mathcal{A} \subset \mathcal{P} : |{\mathcal{A}}|=T} I\left(\bm{y}_{\mathcal{A}}; f_{\mathcal{A}}\right) \,.
\end{equation}
The bounds on information gains $\gamma_{T}$ depend on the kernels used.
For example, $\gamma_{T} = O(d \log T)$ under a linear regression, $\gamma_{T} = O((\log T)^{d+1})$ under a squared exponential kernel, and $\gamma_{T} = O\left(T^{d(d+1)/(2\nu + d(d+1))}(\log T)\right)$ under a Mat\'{e}rn kernel with $\nu > 1$.
We refer to Srinivas at al.\@ \cite{Srinivas2010} for more details.
In our model, we denote the maximum information gain as $\gamma_{T}(k_{d}+k_{c})$, which describes the maximum amount of information that the algorithm could learn about the demand and total claims function.

Assume that both demand and total claims functions $f_{d}, f_{c}$ satisfy the following Assumption \ref{Ass:f}.
It allows us to derive the bound on the cumulative regret w.r.t.\@ the maximum information gain $\gamma_{T}(k_{d}+k_{c})$ in Theorem \ref{theorem_GPBounds}.

\begin{assumption}\label{Ass:f}
Let $f$ be sampled from a GP with kernel ${k}(p, p')$ %and kernel $k(\cdot, p)$ is L-Lipschitz.
and function $f$ is almost surely continuously differentiable.
Consider partial derivatives ${\partial f /\partial p_{j}}$ of this sample path for $j=1,\dots, T$ satisfy the following high probability bound.
For some constants $a', b' > 0$,
\begin{equation*}
    \Pro \left(\sup_{p \in \mathcal{P}} \left| \frac{\partial f }{\partial p_{j}}\right| > J \right) \leq a' e^{-(J/b')^2}\,.
\end{equation*}
\end{assumption}
\begin{theorem}\label{theorem_GPBounds}
Pick $\delta \in (0,1)$ and choose
\begin{equation*}
    \varphi_{t} = 2 \log(2 t^2 \pi^2 / (3\delta))
    + 2 \log\left(t^2 \right)
    \in O(\log t)\,.
\end{equation*}
With high probability $1 - \delta$ and for any $T \geq 1$,
\begin{equation*}
    \regret(T) = O \left(\sqrt{\gamma_{T}T \log T}\right)\,.
\end{equation*}
Here $\gamma_{T}$ is short for $\gamma_{T}(k_{d}+k_{c})$.
\end{theorem}

The proof follows Srinivas et al.\@ \cite{Srinivas2010}.
We extend their model to the case of demands and claims.
We present related lemmas and results from Srinivas et al.\@ \cite{Srinivas2010} in \nameref{Appendix} \ref{Appendix_GPUCB:RegBounds}.
\begin{proof}
For any price $p \in \mathcal{P}$, we consider demand and claims functions $f_{d}, f_{c} \colon \mathcal{P} \to \R$ are sampled from GPs.
%The optimal price ${p}^{*} \in \mathcal{P}$ is the price that maximizes the revenue function $r(p_{t})$.
The cumulative regret over time horizon $T$ is given by
\begin{align*}
    \begin{split}
        \regret(T)
        & =
        \Exp
        \left[
        \sum_{t=1}^{T}
        \left({p}^{*} \cdot f_{d}({p}^{*}) - f_{c}({p}^{*})\right)
        -\left({p}_{t} \cdot f_{d}({p}_{t}) - f_{c}({p}_{t})\right)
        \right] \\
        & = \Exp
        \left[
        \sum_{t=1}^{T} r(p^{*}) -  r(p_{t})\right]\,.
    \end{split}
\end{align*}
%The price set $\mathcal{P}$ is bounded, and expected demand and claims functions $f_{d}, f_{c}$ are independently multivariate Gaussian distributed with bounded variances .
The price set $\mathcal{P}$ is bounded, and expected demand and claims functions $f_{d}, f_{c}$ are independently multivariate Gaussian distributed with bounded variances.
Without loss of generality, we assume that $k_{r} \leq 1$.
By the UCB rule, we can obtain the bound on regret (see Lemma \ref{lemma:simpleregretbound2} in  \nameref{Appendix} \ref{Appendix_GPUCB:RegBounds}), given by
\begin{align}\label{eq1:simpleregret}
    \begin{split}
        r(p^{*}) -  r(p_{t})
        \leq 2\sqrt{\varphi_{t}} {\sigma}_{t-1}^{r}(p_{t})+  \frac{r b  \sqrt{\log(2a/\delta)}}{t^2} \,.
    \end{split}
\end{align}
Here,
${\sigma}_{t-1}^{r}({p}) = p_{t-1}{\sigma}_{t-1}^{d}({p}) + {\sigma}_{t-1}^{c}({p})$.
By the convexity of logarithm function, we know
$u^2\leq v^2\log{(1+u^2)}/ \log{(1+v^2)}$ for $u \leq v$.
Since ${\sigma}_{t-1}^{r}({p_{t}})^{2} \leq k({p}_{t}, {p}_{t})=1$ and ${\sigma}^{-2}{\sigma}_{t-1}^{r}({p_{t}})^{2} \leq {\sigma}^{-2}$, we let $u^{2} = {\sigma}^{-2}{\sigma}_{t-1}^{r}({p_{t}})^{2}$ and $v^{2} = {\sigma}^{-2}$.
Then we have
\begin{align*}
    \begin{split}
        {\sigma}_{t-1}^{r}({p_{t}})^{2}
        & \leq \frac{1}{\log\left(1+{\sigma}^{-2}\right)} \log\left(1+{\sigma}^{-2}{\sigma}_{t-1}^{r}({p_{t}})^{2}\right)\,.
    \end{split}
\end{align*}
By Jensen's inequality and $\varphi_{t}$ is nondecreasing, we have
\begin{align}\label{eq:sumvarphisigma}
    \begin{split}
        \sum_{t=1}^{T} 2 \sqrt{\varphi_{t}} {\sigma}_{t-1}^{r}(p_{t})
        & \leq
        2 \sqrt{\varphi_{T}}\sqrt{T\sum_{t=1}^{T}{\sigma}_{t-1}^{r}(p_{t}) ^{2}} \\
        & \leq
        2 \sqrt{\varphi_{T}}\sqrt{T \sum_{t=1}^{T} \frac{\log\left(1+{\sigma}^{-2}{\sigma}_{t-1}^{r}({p_{t}})^{2}\right)  }{\log\left(1+{\sigma}^{-2}\right)}}\\
        & \leq
        \sqrt{ T \varphi_{T}}
        \sqrt{\frac{8 I(\bm{y}_{T}; f_{T})  }{\log\left(1+{\sigma}^{-2}\right)}}\\
        & \leq
        \sqrt{ C_{1} T \varphi_{T}  \gamma_{T}}\,.
    \end{split}
\end{align}
here $C_{1} = 8/\log\left(1+{\sigma}^{-2}\right)$ and $\gamma_{T} = \gamma_{T}(k_{d}+k_{c})$.
The last inequality is obtained by the definition of information gain $I$ (see Lemma \ref{Lemma:InfoGain} in \nameref{Appendix} \ref{Appendix_GPUCB:RegBounds}) and maximum information gain $\gamma_{T}$ \eqref{eq:gamma}.

By summing up \eqref{eq1:simpleregret} and using \eqref{eq:sumvarphisigma}, we obtain with probability greater than $ 1 - \delta$ for all $T \geq 1$
\begin{align*}
    \begin{split}
        \sum_{t=1}^{T}
        r(p^{*}) -  r(p_{t})
        & \leq  \sum_{t=1}^{T}2
        \left(\sqrt{\varphi_{t}} {\sigma}_{t-1}^{r}({p}_{t})+  \frac{r b  \sqrt{\log(2a/\delta)}}{t^2}\right) \\
        & = \sum_{t=1}^{T}  2\sqrt{\varphi_{t}} {\sigma}_{t-1}^{r}(p_{t})
        + \sum_{t=1}^{T} \left(\frac{r b  \sqrt{\log(2a/\delta)}}{t^2}\right)\\
        & \leq
        \sqrt{ C_{1} T \varphi_{T}  \gamma_{T}} + C_{2} \,.
    \end{split}
\end{align*}
Here, $C_{2} = \sum_{t=1}^{T} \left( {r b \sqrt{\log(2a/\delta)}}/{t^2}\right)$ is a constant since $a, b, r, \delta$ are constants and $\sum 1/t^2 = \pi^2/6$.
Finally, we obtain the bounds on the cumulative regret.
\end{proof}
Note that for the combination of additive kernels, we have
$\gamma_{T}(k_{d}+k_{c}) \leq \gamma_{T}(k_{d}) + \gamma_{T}(k_{c})$ (see Lemma \ref{lemma:InfoGain2kernels} in  \nameref{Appendix} \ref{Appendix_GPUCB:RegBounds} for more details).

%%%%%%%%%%%%%%%%%%%%%%%%%%%%%%%%%%%%%%%%%
%% Delay Case
%%%%%%%%%%%%%%%%%%%%%%%%%%%%%%%%%%%%%%%%%
\section{Delayed Case}\label{section_DelayedCase}
In previous sections, we assume that the insurance premium and total claims are paid at the beginning of each insurance period.
This gives the basic outline of how these methods can be applied in an insurance setting.
However, in the real world, claims are triggered when the insured events happen, thus are not paid out when an insurance product is purchase.
An immediate pricing decision during the delay may give a wrong optimal price and increase the regret.
As a consequence, we involve delayed claims here and consider our pricing problem as a delayed bandit problem, based on the idea of Joulani \cite{Joulani2016}.
 %Vernade et al.\@ \cite{Vernade2018} and Li et al \cite{LichenG2018}.

\subsection{Delayed Models}
We assume that demands are generated and observed immediately when a price is chosen, while claims might be delayed.
It might lead to delayed revenue.
Suppose a price is chosen at time $t$ and denote the corresponding delayed time is $\tau_{t}$.
We assume that $(\tau_{t})_{t \geq 1}$ is an i.i.d. sequence and is independent of prices and claims.
It is possible that $\tau_{t} = 0$ when there is no delay.
Assume there is a maximum waiting time $m$, that the delayed claims can only be observed by time $t + m$ or if $\tau_{t} \leq m$.
Without loss of generality, we assume that all information (demands and claims) can be received by the end of time horizon $T$.

When delay occurs, the insurance company  observes revenue $r_{t}$ at time $t + \tau_{t}$.
Therefore, the next selling price $p_{t}$ is chosen based on
history $\mathcal{H}_{t}$, which is the past information of prices and demands $\{p_i\,, d_i : \, i = 1, \dots, t-1\}$ and delayed claims $\{c_{i}: \, i + \tau_{i} \leq t-1 \}$.
If $\tau_{t} = 0$ for all $t$, we have
$\mathcal{H}_{t} = \mathcal{F}_{t}$.

If a price $p_{t'}$ is chosen at time $t'$, we denote the claims observed at time $t$ $(= t' + \tau_{t'} )$ by
${c}_{t' + \tau_{t'}}$.
%\Msout{It is possible that ${c}_{t' + \tau_{t'}} = 0$, specifically if $\tau_{t'} = m$.}
The insurance company collects the set of start times for delayed claims: $\mathcal{C}_{t} = \{t':\, t' = t - \tau_{t'}\}$ by the end of time $t$.
Here, the delayed time set can be $\mathcal{C}_{t} = \emptyset$ for instance when $\tau_{t'} = 0$ or $\tau_{t'} > m$ for all $t'$.
The corresponding observed revenue at time $t$ is $r_t = p_{t'} d_{t'} - {c}_{t' + \tau_{t'}}$.
It is worth noticing that there is less information to decide $\bm{p}(t+1)$ due to the delay of claims.

We denote the time that the insurance company observes the $s$-th claim as $\rho(s)$ and the corresponding revenue as $r_{\rho(s)}$, here $s = 1, \dots, T$.
Let $\widetilde{r}_{s} = r_{\rho(s)}$.
The insurance company determines the next selling price $\widetilde{p}_{s + 1}$ after observing $\widetilde{r}_{s}$.
For $t = 1, \dots, T$, we denote the number of claims observed by the end of time $t-1$ as $N(t) = \sum_{i=1}^{t-1} \mathbbm{1}\{{i + \tau_{i} < t}\}$.
Since the company determines the next selling price by the latest observed claim, we have $p_{t} = \widetilde{p}_{N(t) + 1}$.

Let $\widetilde{\tau}_{s} = s - 1 - N(\rho(s))$.
Here, $s - 1$ is the number of claims that can be observed if there is no delay.
$N(\rho(s))$ is the actual number of claims that observed by time $\rho(s) -1$.
Therefore, $\widetilde{\tau}_{s}$ is the number of delays that have not been updated when the algorithm chooses $p_{\rho(s)}$ at time $\rho(s)$ until the corresponding revenue $r_{\rho(s)}$ can be observed.
Then we have
$p_{\rho(s)} = \widetilde{p}_{s-\widetilde{\tau}_{s}}$ and $\sum_{s=1}^{T}\widetilde{\tau}_{s} = \sum_{t=1}^{T}{\tau}_{t}$ (for a proof we refer to Lemma \ref{lemma:sumtildetau} in \nameref{Appendix} \ref{Appendix_DelayGLM}).
Hence, the regret is
\begin{align*}\label{eq:regretT}
    \begin{split}
        \sum_{t=1}^{T}
        \sum_{t'\in \mathcal{C}_{t}}
        {r({p}^{*})-r({p}_{t})}
        &=
        \sum_{s=1}^{T} {r_{\rho(s)}({p}^{*})-r_{\rho(s)}({p}_{\rho(s)})}\\
        &=
        \sum_{s=1}^{T} {\widetilde{r}_{s}({p}^{*})-\widetilde{r}_{s}(\widetilde{p}_{s - \widetilde{\tau}_{s}})}\\
        &=
        \sum_{s=1}^{T} \widetilde{r}_s({p}^{*})- \widetilde{r}_{s}(\widetilde{p}_{s})
        +
        \sum_{s=1}^{T}\widetilde{r}_{s}(\widetilde{p}_{s}) - \widetilde{r}_{s}(\widetilde{p}_{s - \widetilde{\tau}_{s}})\,.
    \end{split}
\end{align*}
We can see that the regret with delayed claims has two parts: the non-delayed regret and an additional regret caused by delayed claims. We need to bound each of these two terms to get the overall regret bound.

\subsection{Adaptive GLM Pricing with Unknown Delays}
In GLMs regression model without delays, we find the unknown parameters by maximizing the quasi-likelihood estimation.
In the delayed case, we use a similar method.
We denote the new least squares estimator of delayed claims by $\widehat{\bm{b}}_{t}'$. Specifically $\widehat{\bm{b}}_{t}'$ is the modified MQLE given by
\begin{equation*}
    \bm{l}_{t}(\widehat{\bm{b}}_{t}')
    =
    \sum_{i=1}^{t}
    \sum_{i'\in \mathcal{C}_{i}}
    \frac{1}{\sigma^2}
    \bm{p}(i)
    \left({c}_{i'+\tau_{i'}}-{h \left(\bm{p}^\top(i)\widehat{\bm{b}}_{t}'\right)} \right)
    =\bm{0} \,.
\end{equation*}
The following result extends Theorem \ref{theorem_GLMBound} to delayed claims.

\begin{theorem}\label{theorem_GLMDelayBound}
Suppose there exists $t_{0} \in \N$, if MQLE $\widehat{{\bm{\beta}}}_{t}$ is  strongly consistent and following conditions are satisfied:
\begin{enumerate}
    \item{
    $\lambda_{\min}(t) \geq L_{1}(t) = c \sqrt{t \log(t)}$ a.s. for all $t \geq t_{0}$ and $ c>0$,
    }
    \item \label{DelayUbd_2} {
    $\sum_{t=1}^{T} \left\| \bm{p}(t) - \bm{p}(\widehat{{\bm{\beta}}}_{t}) \right\|^2 \leq \left(2m+K\right) L_{1}(T)$ a.s.\@ for all $T \geq t_{0}$ and some $K_1 >0$.
    }
\end{enumerate}
The regret bound is
\begin{align*}
    \regret\,(T)
    & =
    O \left(L_{1}(T)+ \sum_{t=1}^{T} \frac{\log(t)}{L_{1}(t)}  \right)
\end{align*}
\end{theorem}
\begin{proof}
Proof is similar to  Theorem \ref{theorem_GLMBound}.
We consider the cumulative regret over selling horizon $T$ in terms of $\left\|\bm{p}(t)-\bm{p}^{*} \right\|^{2}$.
This is
\begin{align*}
    \regret(T)
    & =
    O \left(
    \Exp\left[\sum_{t=t_{0}}^{T} \left\|\bm{p}(t)-\bm{p}^{*} \right\|^{2}\right]
    \right) \,.
\end{align*}
Here, the expected value of $\left\|\bm{p}(t)-\bm{p}^{*} \right\|^{2}$ is given by
\begin{align*}
   \begin{split}
    \Exp\left[\sum_{t=t_{0}}^{T} \left\|\bm{p}(t)-\bm{p}^{*} \right\|^{2}\right]
    & \leq 2 \Exp\left[\sum_{t=t_{0}}^{T} \left\|\bm{p}(t)-\bm{p}(\widehat{\bm{\beta}}_{t}) \right\|^{2} \right]
    +
    2 \Exp\left[\sum_{t=t_{0}}^{T} \left\|\bm{p}(\widehat{\bm{\beta}}_{t}) - \bm{p}^{*}\right\|^{2} \right] \\
    & \leq 2 \left(2m
    + K \right) L_{1}(T)
    +
    2 K_{2} \Exp\left[\sum_{t=t_{0}}^{T} \left\|\widehat{\bm{\beta}}_{t} -  {\bm{\beta}}_{0}\right\|^{2} \right] \\
    & = O \left(L_{1}(T)+ \sum_{t=1}^{T}
    \left(
    \frac{\log(t)}{L_{1}(t)}
    \right)
    \right)
   \end{split}
\end{align*}
In the first inequality, we use
$(a+b)^2 \leq 2a^2 +2b^2$, that is
\begin{align*}
    \begin{split}
        \left\|\bm{p}(t)-\bm{p}^{*} \right\|^{2}
        & = \left\|\bm{p}(t)-\bm{p}(\widehat{\bm{\beta}}_{t}) + \bm{p}(\widehat{\bm{\beta}}_{t}) - \bm{p}^{*} \right\|^{2} \\
        & \leq
        2  \left\|\bm{p}(t)-\bm{p}(\widehat{\bm{\beta}}_{t}) \right\|^{2} +
        2 \left\|\bm{p}(\widehat{\bm{\beta}}_{t}) - \bm{p}^{*}\right\|^{2} \,.
    \end{split}
\end{align*}
We extend the first term in the above inequality, that is,
\begin{align*}
    \begin{split}
        \sum_{t=1}^{T} \left\|\bm{p}(t)-\bm{p}(\widehat{\bm{\beta}}_{t}) \right\|^{2}
        & \leq
        \sum_{s=1}^{T} \left\|\widetilde{\bm{p}}({s - \widetilde{\tau}_{s}})-
        \widetilde{\bm{p}}(s)
        + \widetilde{\bm{p}}(s)- \bm{p}( \widehat{\bm{\beta}}_{t}) \right\|^{2} \\
        & \leq
        \sum_{s=1}^{T} \left\|\widetilde{\bm{p}}({s - \widetilde{\tau}_{s}})-
        \widetilde{\bm{p}}(s)\right\|^{2}
        + \sum_{s=1}^{T} \left\|\widetilde{\bm{p}}(s)- \bm{p}( \widehat{\bm{\beta}}_{t}) \right\|^{2}\,.
    \end{split}
\end{align*}
Here by the algorithm, we can obtain
\begin{align*}
    \begin{split}
        \left\|\widetilde{\bm{p}}({s - \widetilde{\tau}_{s}})-
        \widetilde{\bm{p}}(s)\right\|^{2}
        & \leq
        \sum_{j = 0}^{\widetilde{\tau}_{s} -1}
        \left\|\widetilde{\bm{p}}({s - \widetilde{\tau}_{s}} + j)-
        \widetilde{\bm{p}}(s - \widetilde{\tau}_{s} + j + 1) \right\|^{2}\\
        & \leq \widetilde{\tau}_{s} {K}{\dot{L}_{1}(t)}
        \,,
    \end{split}
\end{align*}
for some ${K} > 0$.
Since $L_{1}(t)$ is an increasing concave function, by mean value theorem we have $L_{1}(t+1) - L_{1}(t) = \dot{L}_{1}(s)$ for any $s \in [t,t+1]$.
It implies $L_{1}(t+1) - L_{1}(t) \geq \dot{L}_{1}(t+1)$ and
\begin{equation*}
    \sum_{t=0}^{T-1} \dot{L}_{1}(t+1) \leq \sum_{t=0}^{T-1} L_{1}(t+1) - L_{1}(t) =  L_{1}(T) - L_{1}(0)
    = L_{1}(T)\,.
\end{equation*}
Together with Assumption \ref{DelayUbd_2} in Theorem \ref{theorem_GLMDelayBound} and $\widetilde{\tau}_{s} \leq 2m$ for any $s$ (for a proof we refer to Lemma \ref{lemma:tildetau} in  \nameref{Appendix} \ref{Appendix_DelayGLM}), we obtain the bound on the first term in the second inequality,
\begin{align*}
    \begin{split}
        \Exp\left[\sum_{t=t_{0}}^{T} \left\|\bm{p}(t)-\bm{p}(\widehat{\bm{\beta}}_{t}) \right\|^{2} \right]
        & \leq 2m L_{1}(T)  + K_1 L_{1}(T) \,,
    \end{split}
\end{align*}
for some $K_1 > 0$.
By Proposition \ref{proposition_GLMboundsClinkF},  the bound on the second term is given by
\begin{align*}
    \Exp\left[\left\|\widehat{{\bm{\beta}}}_{t}-{\bm{\beta}}_{0}\right\|^2\right]
    &= O \left(\frac{\log(t)}{L_{1}(t)}\right) \, .
\end{align*}
Finally, the regret bound is obtained.
\end{proof}

\subsection{Adaptive GP Pricing with Unknown Delays}
Currently, there do not exist theoretical bound for GP with delays.
In this section, we present the implementation of an alternative GP pricing algorithm for insurance in a delayed-claim setting.
%\Msout{We use the GP Algorithm~\ref{algo_GPPricing} to choose the next selling price only after delayed claims are observed.}
For example, at each time $t$, we set a price by GP Algorithm~\ref{algo_GPPricing} and collect the set of start times for delayed claims: $\mathcal{C}_{t} = \{t':\, t' = t - \tau_{t'}\}$.
For each time $t' \in \mathcal{C}_{t}$, the premium can be observed at time $t'$ given as ${p}_{t'} \cdot f_{t'}^{d}({p}_{t'})$, while claims are received at time $t'+\tau_{t'}$.
We denote the delayed claims by $f_{t'+\tau_{t'}}^{c}({p}_{t'})$.
After receiving the delayed claims, we update the GP for each claim and its premium. This then determines the next selling price.
The pseudocode for pricing a new released insurance product with delays via GP algorithm is shown in Algorithm~\ref{algo_DelayGPPricing}.
% Algorithm3
\begin{algorithm}[ht]
\caption{GP Pricing Algorithm with Delayed Claims}
\label{algo_DelayGPPricing}
\KwIn{GP Pricing Algorithm~\ref{algo_GPPricing}}

\For{$t = 1, 2, \dots$}{
    Collect the set of delay time $\mathcal{C}_{t}$ received by time $t$ \;
    \For{$t' \in \mathcal{C}_{t}$}{
        Select price:
        ${p}_{t'} \gets$ GP Pricing Algorithm~\ref{algo_GPPricing} \;
        %Sample the revenue function: ${r}_{t'} \gets {p}_{t'} \cdot f_{d}({p}_{t'}) - f_{c}({p}_{t'+\tau_{t'}})$ \;
        Update GP with ${p}_{t'}, f_{t'}^{d}({p}_{t'}), f_{t'+\tau_{t'}}^{c}({p}_{t'})$:
        }
    }
\end{algorithm}

%%%%%%%%%%%%%%%%%%%%%%%%%%%%%%%%%%%%%%%%%
%% Numerical Examples
%%%%%%%%%%%%%%%%%%%%%%%%%%%%%%%%%%%%%%%%%
\section{Explicit Formulas and Numerical Examples} \label{section_NumericalExamples}
In this section, we present the simulation result of the GLMs and GP pricing algorithms and both with delayed claims in numerical examples.

\subsection{Adaptive GLM Pricing Without Delays}
Assume demand follows logistics distribution and total claims follow lognormal distribution under canonical link functions.
For example, let the demand and claims model be
\begin{align*}
    \begin{split}
        \Exp[D(p)]&=11 - 0.8 p\, ,\\
        \Exp[\log C(p)]&=3 + 0.25 p\, .
    \end{split}
\end{align*}
\noindent
Set three initial price vectors to be $\bm{p}(1) = (1, 3)^{\top}$,
$\bm{p}(2) = (1, 3.3)^{\top}$ and
$\bm{p}(3) = (1, 4.7)^{\top}$,
with the lowest and highest price
${p}_{l} = 0.5$ and ${p}_{h} = 10$.
We now go about estimating the parameters of the above optimization.

We plot price dispersion and convergence of parameter estimates.
According to Figure \ref{fig:Pricedispersion}, the price dispersion $\tr(P_{t}^{-1})^{-1}/\sqrt{t \log(t)}$ converges to 0.01.
Thus, we set $L_{1}(t) = 0.01 \sqrt{t \log(t)}$, which can guarantee sufficient price dispersion.
Figure \ref{fig:BataConvergence} presents $\left\| \widehat{\bm{\beta}}_{t} - {\bm{\beta}}_{0}\right \|^2$, the squared norm of the difference between the parameter estimates and the true parameters.
As $t \to \infty$, $\left\| \widehat{\bm{\beta}}_{t} - {\bm{\beta}}_{0}\right \|^2 \to 0$, that is, the parameter estimates converge to the true parameters.
It implies the strong consistency of parameter estimates.
%Regret is the expected loss caused by not using optimal price.
%The convergence of regret in Figure \ref{fig:CumRegretGLM} also shows the strong convergence of parameter estimates. Figure \ref{fig:RegretRateGLM} illustrates the bound on $\regret(T)$, that is, $\regret(T)=O(\sqrt{T\log(T)})$ given an optimal pricing strategy $\psi$.
\begin{figure}[t]
    \begin{center}
    \centering
     \begin{subfigure}[b]{0.48\textwidth}
        \includegraphics[width=\textwidth]{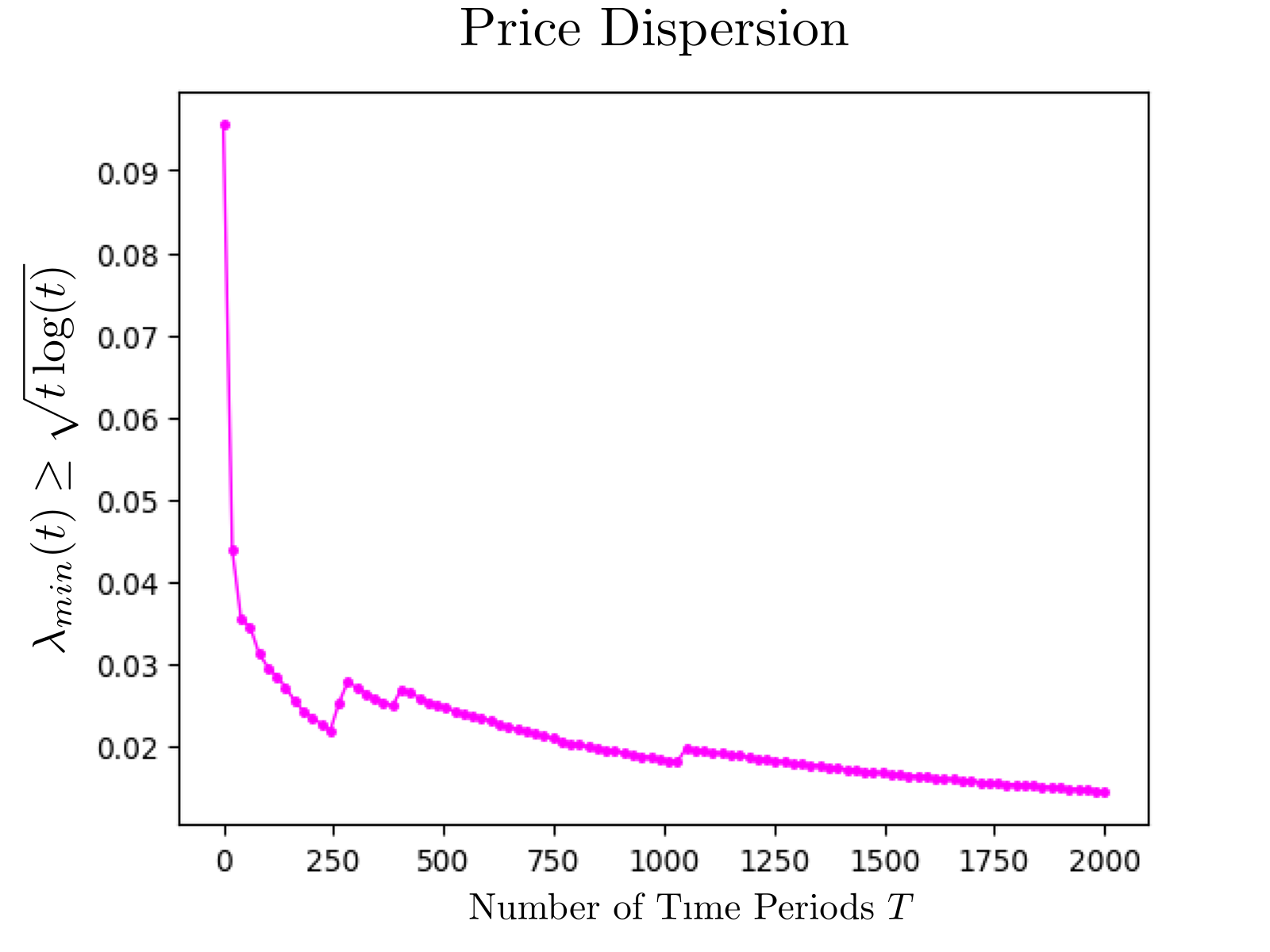}
        \caption{}
        \label{fig:Pricedispersion}
    \end{subfigure}
    \hfill
    \begin{subfigure}[b]{0.48\textwidth}
        \includegraphics[width=\textwidth]{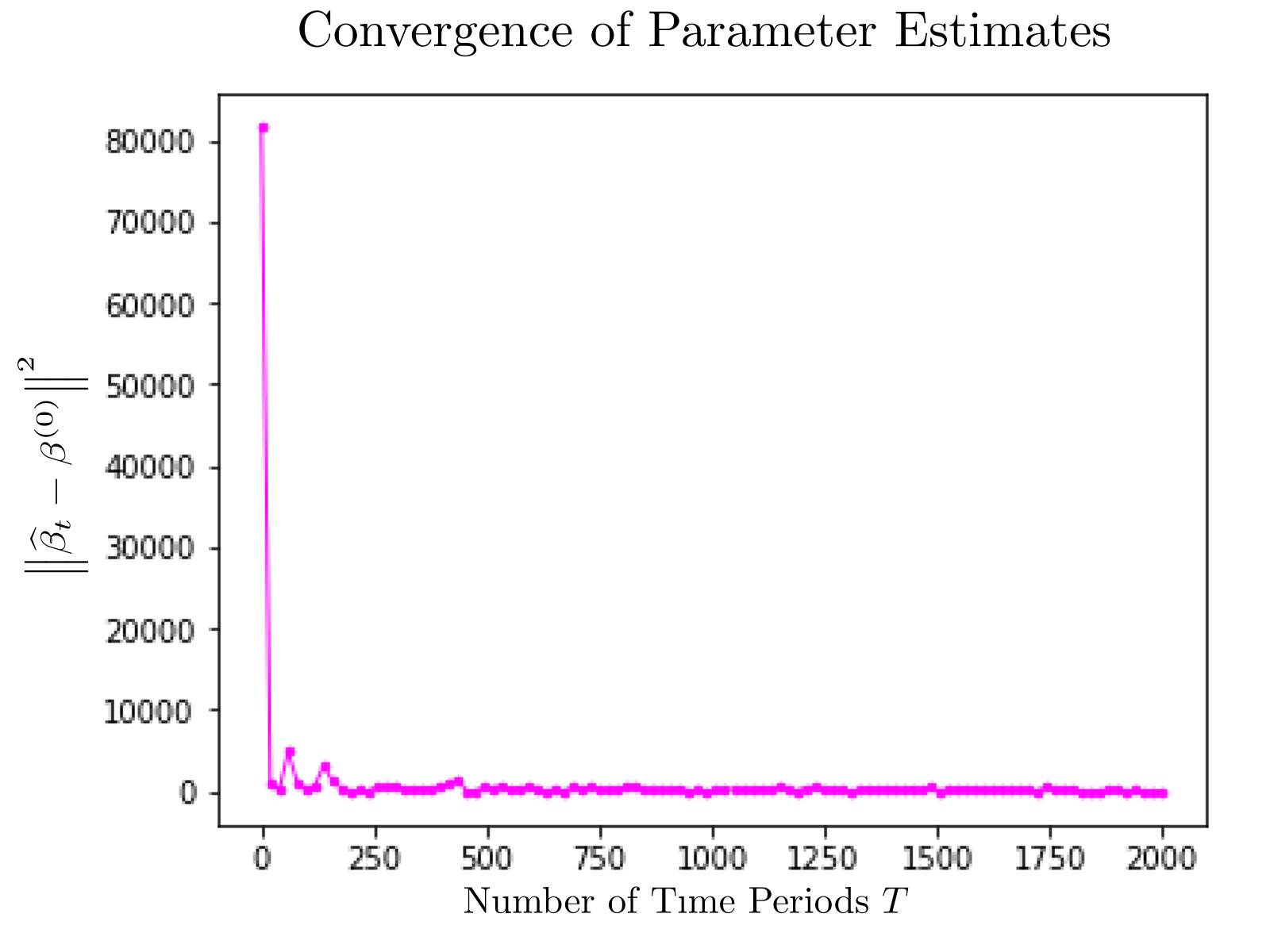}
        \caption{}
        \label{fig:BataConvergence}
    \end{subfigure}
    \caption{Price dispersion and convergence of parameter estimates.}
    \end{center}
\end{figure}
\iffalse
\begin{figure}[t]
     \centering
     \begin{subfigure}[b]{0.48\textwidth}
         \centering
         \includegraphics[width=\textwidth]{MQLE_RgeretSumT.png}
         \caption{Cumulative regret}
         \label{fig:CumRegretGLM}
     \end{subfigure}
     \hfill
     \begin{subfigure}[b]{0.48\textwidth}
         \centering
         \includegraphics[width=\textwidth]{MQLE_RgeretRate.png}
         \caption{Convergence rate}
         \label{fig:RegretRateGLM}
     \end{subfigure}
    \caption{Cumulative regret and convergence of regret.}
    \label{fig:RegretGLM}
\end{figure}
\fi
\subsection{Adaptive GP Pricing Without Delays}
\noindent
The most interesting kernels for machine learning are Mat\'{e}rn kernels with $\nu = 3/2$ and $\nu = 5/2$, given by
\begin{align*}
    \begin{split}
        k_{3/2}(r)
        & =\left(1+ \frac{\sqrt{3} r }{l}\right) \exp \left(-\frac{\sqrt{3} r }{l}\right)  \,,\\
        k_{5/2}(r)
        & =\left(1+ \frac{\sqrt{5} r }{l} +\frac{\sqrt{5} r^2 }{3 l^2} \right) \exp \left(-\frac{\sqrt{5} r }{l}\right) \,.
    \end{split}
\end{align*}
We sample random functions $f_{d}, f_{c}$ from GPs with Mat\'{e}rn kernels with $\nu = 1.5, 2.5$ and length scale $l=1$, respectively.
Set sample noise to be $\sigma = 0.05$.
The GP algorithm runs for $T = 100$ iterations with $\delta = 0.1$.
This indicates that our choice of the parameter $\bm{\beta}$ depends on $t$.
%Figure \ref{fig:RegretGPUCB} shows the cumulative regret and asymptotic optimal regret, $\regret(T)/\sqrt{T \log(T)}$, incurred by the GP pricing algorithm.
\iffalse
\begin{figure}[t]
     \centering
     \begin{subfigure}[b]{0.48\textwidth}
         \centering
         \includegraphics[width=\textwidth]{GPUCB_RegretSumNew.png}
         \caption{Cumulative regret}
         \label{fig:CumRegretGPUCB}
     \end{subfigure}
     \hfill
     \begin{subfigure}[b]{0.48\textwidth}
         \centering
         \includegraphics[width=\textwidth]{GPUCB_RegretRateNew.png}
         \caption{Convergence rate}
         \label{fig:RegretRateGPUCB}
     \end{subfigure}
    \caption{Cumulative regret and rate of convergence }
    \label{fig:RegretGPUCB}
\end{figure}
\fi

\subsection{Delay cases}
In order to find the effects of delays, we consider same settings as those in non-delayed cases.
Delays are non-negative integers and unknown, and we generate delays randomly.
\subsection{Compare GLM and GP Algorithms Without and With Delays}
We present the numerical results in Figure \ref{fig:DelayRegretGLM}  and Figure \ref{fig:DelayRegretGPUCB}.
The cumulative regret incurred by the GLM pricing algorithm is plotted in \ref{fig:DelayCumRegretGLM}.
The regret in delayed setting is always larger than that in non-delayed setting, which suggests that the insurance company suffers from an extra loss.
As we mentioned above, in the delayed case, there is less information to decide selling prices due to the delay of claims.
Figure \ref{fig:DelayRegretGLM} illustrates the regret  converges at the rate of $1/\sqrt{T \log(T)}$.
The convergence rate of regret with delayed claims is asymptotic of the rate without delayed claims.

The results obtained by GP pricing algorithm in Figure \ref{fig:DelayRegretGPUCB} is similar to GLM pricing algorithm.
\ref{fig:DelayCumRegretGPUCB} shows that
regret in delayed setting is larger than the regret in non-delayed setting and \ref{fig:DelayRegretRateGPUCB} shows that convergence rate of regret with delayed and without delayed claims are asymptotic.

By comparing Figure \ref{fig:DelayCumRegretGLM} and Figure \ref{fig:DelayCumRegretGPUCB}, we can see that the cumulative regret obtained by GP pricing algorithm converges faster and gives smaller regret.
Figure \ref{fig:DelayRegretRateGLM} and Figure \ref{fig:DelayRegretRateGPUCB} illustrate the same results that both algorithms achieve same convergence rate, that is, $1/\sqrt{T \log(T)}$.
Overall, the GP pricing algorithm outperforms GLM pricing algorithm.

\begin{figure}[t]
     \centering
     \begin{subfigure}[b]{0.48\textwidth}
         \centering
         \includegraphics[width=\textwidth]{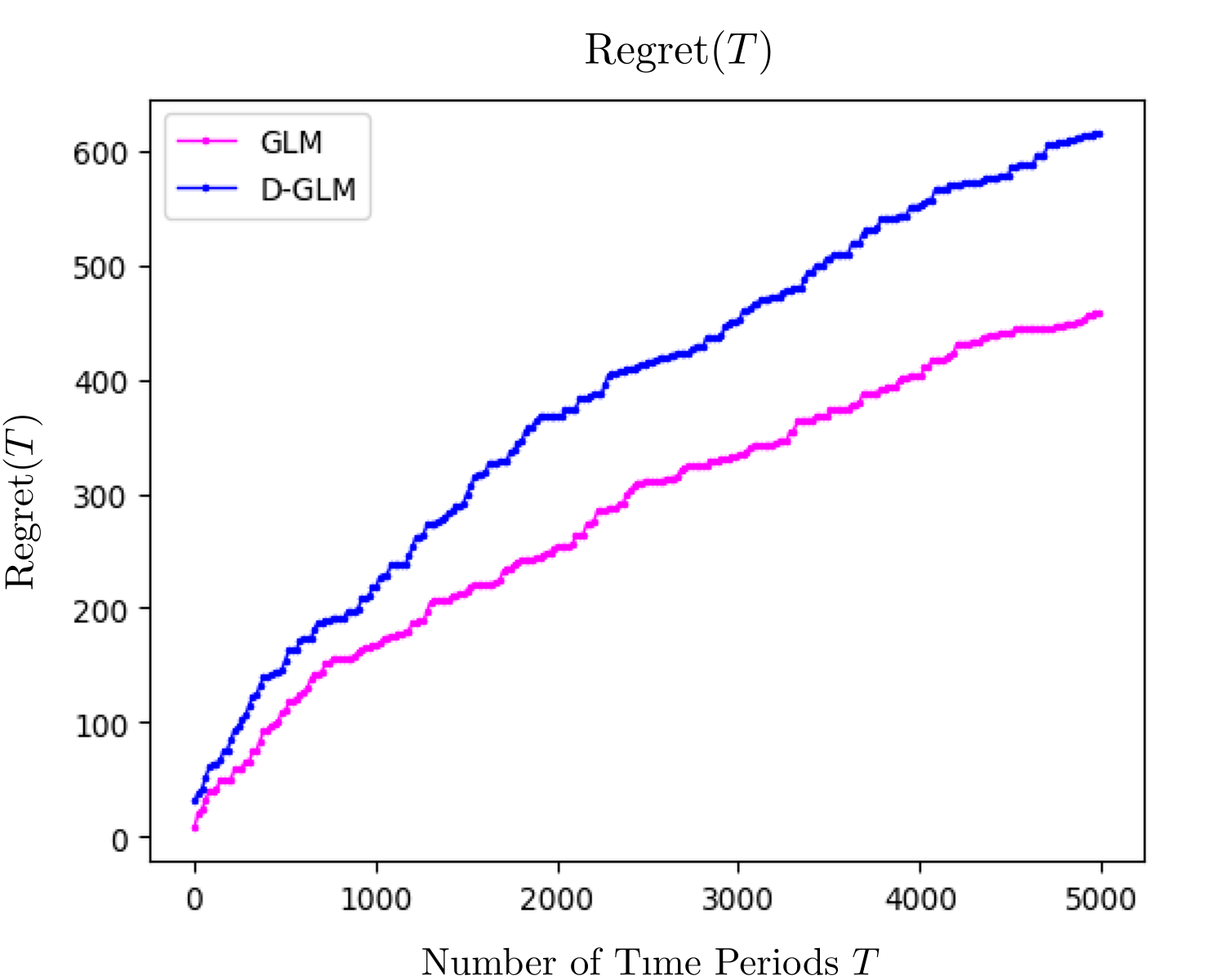}
         \caption{Cumulative regret}
         \label{fig:DelayCumRegretGLM}
     \end{subfigure}
     \hfill
     \begin{subfigure}[b]{0.48\textwidth}
         \centering
         \includegraphics[width=\textwidth]{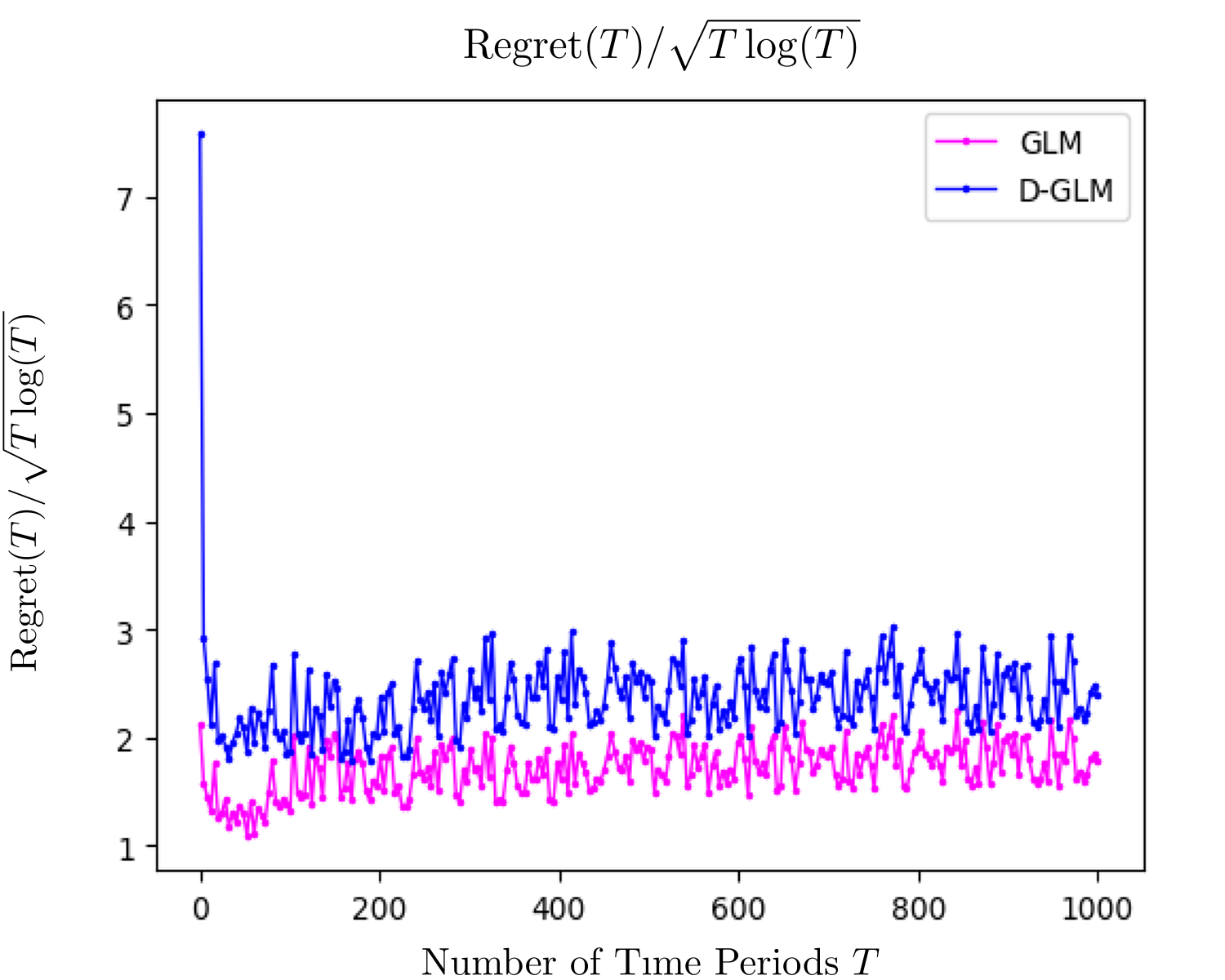}
         \caption{Rate of convergence}
         \label{fig:DelayRegretRateGLM}
     \end{subfigure}
    \caption{Cumulative regret and convergence rate for GLM algorithm. GLM denotes the non-delayed case and D-GLM denotes the delayed case. }
    \label{fig:DelayRegretGLM}
\end{figure}
\begin{figure}[t]
     \centering
     \begin{subfigure}[b]{0.48\textwidth}
         \centering
         \includegraphics[width=\textwidth]{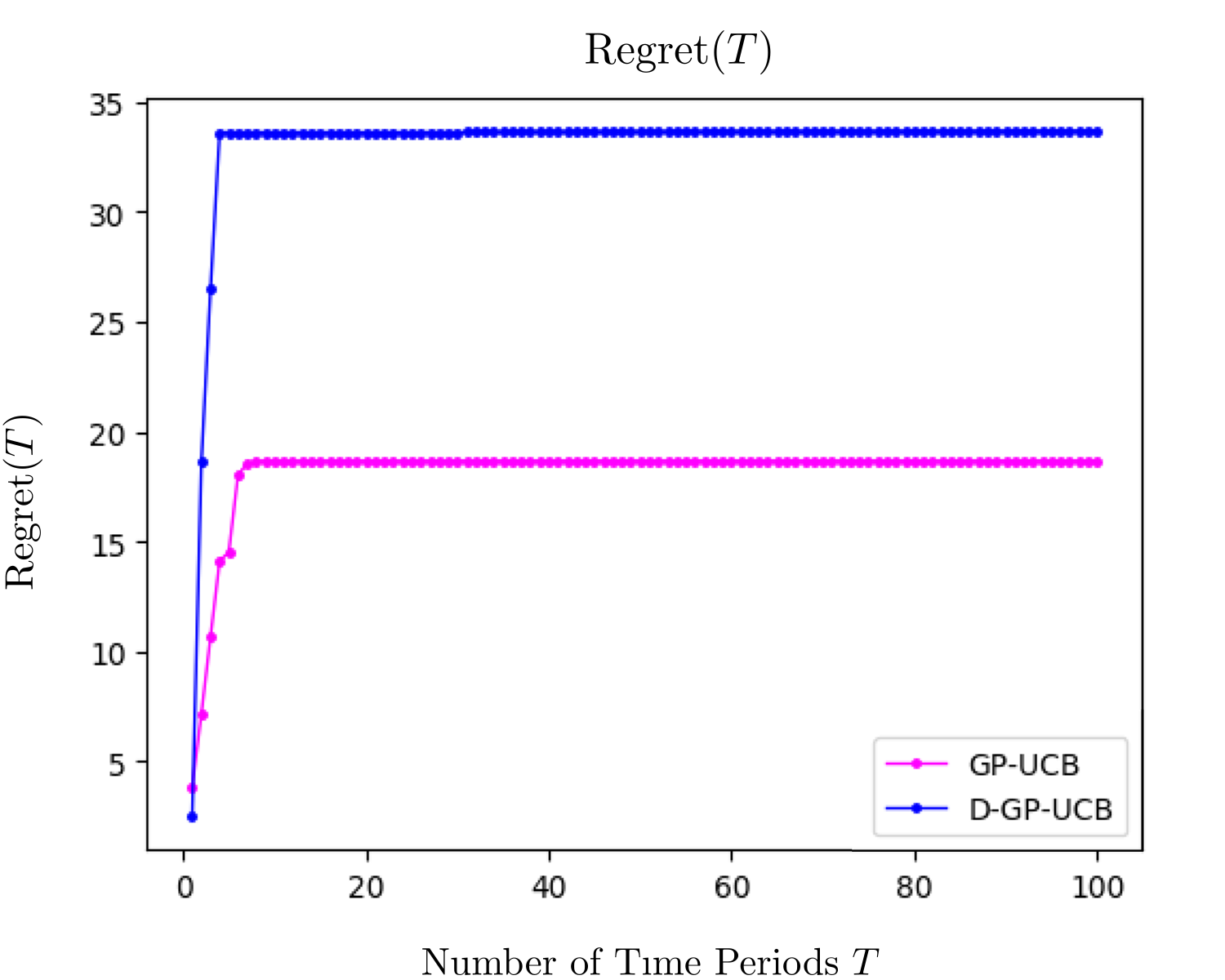}
         \caption{Cumulative regret}
         \label{fig:DelayCumRegretGPUCB}
     \end{subfigure}
     \hfill
     \begin{subfigure}[b]{0.48\textwidth}
         \centering
         \includegraphics[width=\textwidth]{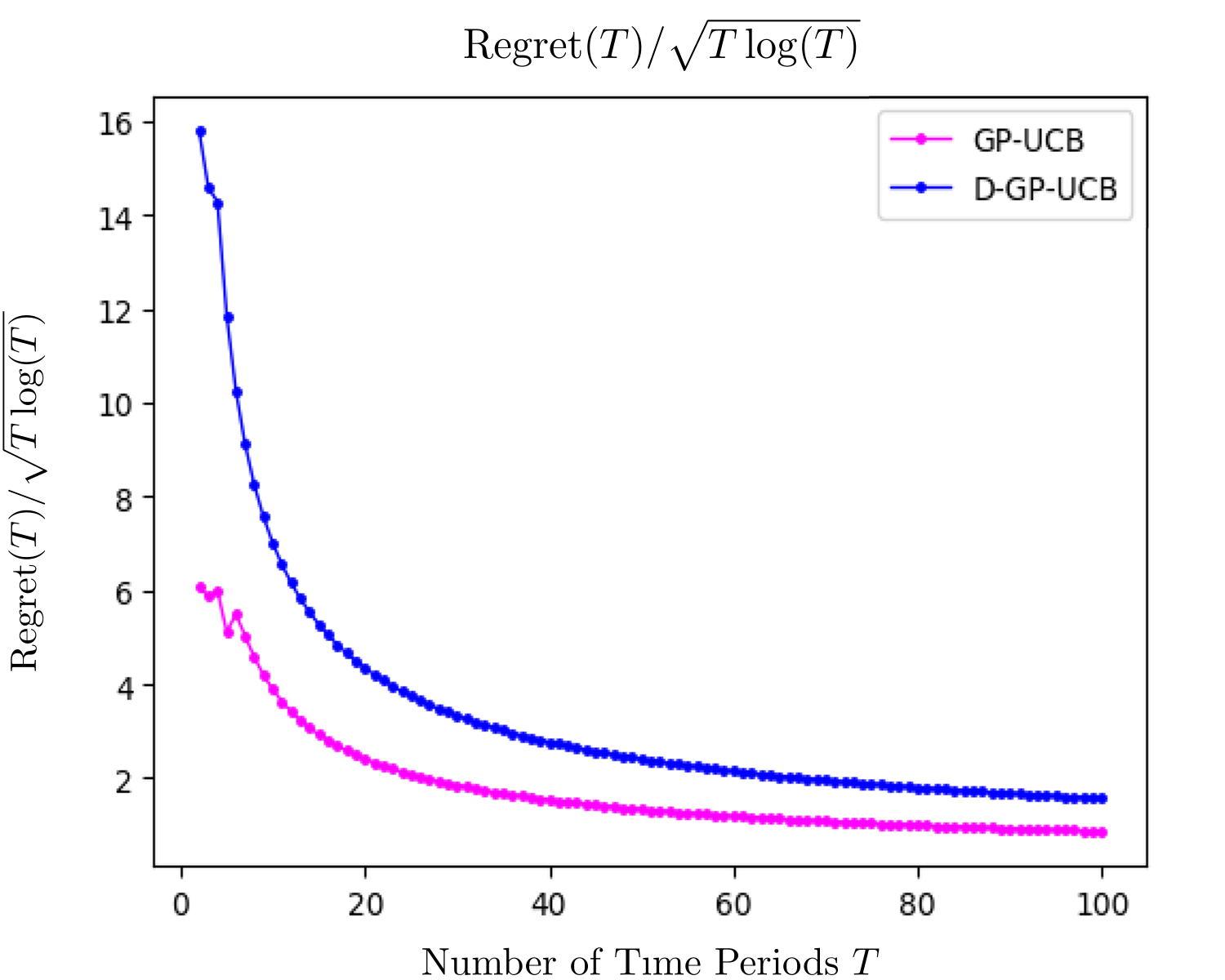}
         \caption{Rate of convergence}
         \label{fig:DelayRegretRateGPUCB}
     \end{subfigure}
    \caption{Cumulative regret and convergence rate for GP algorithm. GP denotes the non-delayed case and D-GP denotes the delayed case. }
    \label{fig:DelayRegretGPUCB}
\end{figure}
%%%%%%%%%%%%%%%%%%%%%%%%%%%%%%%%%%%%%%%%%
%% Conclusion
%%%%%%%%%%%%%%%%%%%%%%%%%%%%%%%%%%%%%%%%%
\section{Conclusion and Discussion}
\label{section_Conclusion}
We considered dynamic learning and pricing problems in an insurance context.
Our work is among the first to apply online learning to insurance.
We found both GLM and GP pricing models have good performance.
Applying prior work to an insurance setting with both demand and claims, we demonstrate this theoretically with precise bounds on regret.
In numerical results, we found GP pricing has better convergence; however, GLMs have a long history of suitable implementation in insurance. Thus it is important to also consider this setting.
%In the  model, we applied Generalized Linear Models (GLM) with maximum quasi-likelihood estimation (MQLE).
%Our experiment shows MQLE parameters eventually exist and converge to the correct values, therefore the sequence of chosen prices also converge to the optimal price.
%In the  pricing model, we used Gaussian Process upper confidence bound (GP) algorithm.
%This algorithm also achieves asymptotically optimal regret .
%Additionally, by comparing the two algorithms we can find that the GP pricing algorithm outperforms the GLM pricing algorithm in terms of the speed of convergence.
More broadly, our findings suggest that the new Gaussian Process regression is potentially applicable in insurance.
However it is currently under studied.

%The company can not observe claims and immediate revenue because claims incur when insured events happen.
%
Since an insurance company cannot immediately observe claims,
in Section \ref{section_DelayedCase}, we extended the GLM pricing model algorithm to have delayed claims. %and presented theoretical bounds on regret.
%
%We assume that all demands and claims  can be observed by the end of time horizon.
We prove that the asymptotic regret bound with delayed claims is the same as that achieved in the non-delayed case. The same result for the GP algorithm remains open, but numerical results suggest the same regret bound is achieved.

There are several ideas for future investigation.
Thus far we focus on the long-run revenue.
In the insurance market, claims may cause loss and even bankruptcy to the insurance company.
One idea is to incorporate ruin probability, which is a measure for the risk used for decision taking.
%as ruin can occur when claims take place.

Another possible idea is to implement reinforcement learning (RL) with these online revenue management problems.
Bandit problems, as considered in this paper, are simple reinforcement learning routines. However, more complex future market interactions might be considered.
This problem may be modeled as a Markov Decision Process (MDP).
RL iteratively interacts with a simulation of the insurance model and then use the feedback from the environment to select actions that maximize the insurer's objective.
Although MDPs have been applied for many years in insurance, extending the online learning framework considered here to incorporate forward planning is an area that is yet to be studied.

\section*{Acknowledgements}\label{Acknowledgements}
This work was supported by China Scholarship Council (CSC) Grant.

\section*{Appendix}\label{Appendix}
\addcontentsline{toc}{section}{Appendices}
\renewcommand{\thesubsection}{\Alph{subsection}}

\subsection{Proof of Proposition \ref{proposition_PropGLMPricingNext}}\label{Proof_PropGLMPricingNext}
Recall from Algorithm~\ref{Algo_GLMPricing}, if ${\widehat{\bm{\beta}}_{t}}$ exists and $\tr(P_{t}^{-1})^{-1}\geq L_{1}(t)$,
we first set the next price to be $\bm{p}(t+1) =  \bm{p}_{cep}$.
If the condition \eqref{eq:TraceL1t} does not hold, that is
\begin{equation}\label{eq:traceLt}
    \tr\left(\left({P_{t}+ \bm{p}_{cep} \bm{p}_{cep}^{\top}}\right)^{-1}\right)^{-1} <  L_{1}(t+1) \,,
\end{equation}
we then choose the next price to be
$\bm{p}{'}=  \bm{p}_{cep} + \bm{\phi}_{t}$.
We will show there exists $\bm{\phi}_{t}$, such that
\begin{equation}\label{eq:traceLt1}
    \tr\left(\left({P_{t}+ \bm{p}{'} \bm{p}{'}^{\top}}\right)^{-1}\right)^{-1}
    \geq L_{1}(t+1) \, ,
\end{equation}
is satisfied.
By the Sherman-Morrison formula in Bartlett \cite{Bartlett1951},
we have
\begin{align*}
    \begin{split}
        \left(P_{t} + \bm{p}{'}\bm{p}{'}^{\top}\right)^{-1}
        & = P_{t}^{-1} - \frac{P_{t}^{-1}\bm{p}{'}\bm{p}{'}^{\top} P_{t}^{-1}}{1 + \bm{p}{'}^{\top} P_{t}^{-1}\bm{p}{'}} \,.
    \end{split}
\end{align*}
Then, we can derive
\begin{align}
    \begin{split}
        \tr\left(\left(P_{t} + \bm{p}{'}\bm{p}{'}^{\top}\right)^{-1} \right)
        & = \tr\left( P_{t}^{-1} - \frac{P_{t}^{-1}\bm{p}{'}\bm{p}{'}^{\top} P_{t}^{-1}}{1 + \bm{p}{'}^{\top} P_{t}^{-1}\bm{p}{'}}\right)  \\
        & = \tr\left( P_{t}^{-1}\right) -   \tr\left( \frac{P_{t}^{-1}\bm{p}{'}\bm{p}{'}^{\top} P_{t}^{-1}}{1 + \bm{p}{'}^{\top} P_{t}^{-1}\bm{p}{'}}\right)  \\
        & = \tr\left( P_{t}^{-1}\right) -   \frac{\tr\left(P_{t}^{-1}\bm{p}{'}\bm{p}{'}^{\top} P_{t}^{-1}\right)}{1 + \bm{p}{'}^{\top} P_{t}^{-1}\bm{p}{'}}
        \\
        & = \tr\left( P_{t}^{-1}\right)
        -   \frac{\left\|P_{t}^{-1}\bm{p}{'}\right\|}{1 + \bm{p}{'}^{\top} P_{t}^{-1}\bm{p}{'}} \\
        & \leq \frac{1}{L_{1}(t)} + \frac{\df}{\df t} \left(\frac{1}{L_{1}(t)}\right)
        \,.
    \end{split}
\end{align}
The last inequality is obtained because $t \to \frac{1}{L_{1}(t)}$ is convex and $\frac{1}{L_{1}(t+1)} \geq \frac{1}{L_{1}(t)} + \frac{\df}{\df t} \left(\frac{1}{L_{1}(t)} \right)$.
Now we want to prove that
\begin{equation}
    \frac{\tr\left(P_{t}^{-1}\bm{p}{'}\bm{p}{'}^{\top} P_{t}^{-1}\right)}{1 + \bm{p}{'}^{\top} P_{t}^{-1}\bm{p}{'}}
    \geq
    - \frac{\df}{\df t} \left(\frac{1}{L_{1}(t)}\right)
\end{equation}

Here, we discuss a general case.
Let $t > n+1$ and let $\lambda_{1} \geq \dots \geq \lambda_{n+1} > 0$ be the eigenvalues of $P_{t} \subset \mathbb{R}^{n+1}$, and $\bm{v}_{1}, \dots, \bm{v}_{n+1}$ be associated eigenvectors. %where $\|\bm{v}_{n+1}\| \leq 1$.
Since $P_{t}$ is a symmetric positive definite matrix and $\bm{v}_{1}, \dots, \bm{v}_{n+1}$ are an orthonormal basis of $\mathbb{R}^{n+1}$,
we can define the optimal price $\bm{p}_{cep} = \sum_{i=1}^{n
+1}\alpha_{i}\bm{v}_{i}$ as a linear combination of these unit eigenvectors.
Let the next price be $\bm{p}{'} = \bm{p}_{cep} + \epsilon \left(v_{n+1,1}\bm{p}_{cep} - \bm{v}_{n+1} \right)$, here $v_{n+1,1}$ is the first component of $\bm{v}_{n+1}$.
We know that $\|\bm{v}_{i}\| = 1$ and $|v_{n+1,i}|\leq 1$ for all $i$.
Then we have
\begin{align*}
    \begin{split}
        \|\bm{p}{'}- \bm{p}_{cep} \|^2
        = \epsilon^2 \| v_{n+1,1}\bm{p}_{cep} - \bm{v}_{n+1} \|^2
        \leq \epsilon^2  \left(1+\max_{\bm{p}\in \mathcal{P}}\|\bm{p}\|^{2}\right)\,.
    \end{split}
\end{align*}
It demonstrates that
\begin{equation}\label{eq:boundphi}
    \|\bm{\phi}_{t}\|^2 \leq \epsilon^2  \left(1+\max_{\bm{p}\in \mathcal{P}}\|\bm{p}\|^{2}\right)\,.
\end{equation}
We choose $|\epsilon| \leq 1$ such that
\begin{align}\label{eq:epsilon}
    \begin{split}
        \epsilon \geq 0 \text{ if } \alpha_{n+1} \leq 0 \,, \epsilon < 0 \text{ if } \alpha_{n+1} > 0 \,,
    \end{split}
\end{align}
and
\begin{align}\label{eq:LvsEpsilon}
    \begin{split}
        \dot{L}_{1}(t)
        \leq \epsilon^2  (n+1)^{-2}\left(1 + L_{1}(n+1)^{-1} \max_{\bm{p}\in \mathcal{P}}\|\bm{p}\|^{2}\right)^{-1} \,.
    \end{split}
\end{align}
We write $\|\bm{p}{'}\|_{P_{t}^{-1}}^2 = \bm{p}{'}^{\top} P_{t}^{-1}\bm{p}{'}$.
Since $\lambda_{\max} \left( P_{t}^{-1}\right) = \lambda_{\min} \left(P_{t}\right)^{-1}$,  $\lambda_{\min} \left( P_{t}\right)\geq L_{1}(t)$ and $t > n+1$, we have
\begin{align}\label{dividend}
    \begin{split}
        1 + \|\bm{p}{'}\|_{P_{t}^{-1}}^2
        & \leq 1 + \lambda_{\max} \left( P_{t}^{-1}\right) \|\bm{p}{'}\|^{2} \\
        & = 1 + \lambda_{\min} \left(P_{t}\right)^{-1} \|\bm{p}{'}\|^{2}\\
        & \leq 1 + L_{1}(t)^{-1} \|\bm{p}{'}\|^{2} \\
        & \leq 1 + L_{1}(n+1)^{-1} \max_{\bm{p}\in \mathcal{P}}\|\bm{p}\|^{2} \,.
    \end{split}
\end{align}
Moreover,
\begin{align}\label{eq:normPtp}
    \begin{split}
        \left\|P_{t}^{-1}\bm{p}{'}\right\|^2
        & = \left\|P_{t}^{-1}\left(\sum_{i=1}^{n+1}\alpha_{i}\bm{v}_{i} + \epsilon \left(v_{n+1,1}\left(\sum_{i=1}^{n+1}\alpha_{i}\bm{v}_{i}\right) - \bm{v}_{n+1} \right)\right)\right\|^2\\
        & = \left\|P_{t}^{-1}\left(\sum_{i=1}^{n+1}\alpha_{i}\bm{v}_{i} + \epsilon \left(v_{n+1,1}\left(\sum_{i=1}^{n}\alpha_{i}\bm{v}_{i} + \alpha_{n+1}\bm{v}_{n+1}\right) - \bm{v}_{n+1} \right)\right)\right\|^2\\
        & = \left\|P_{t}^{-1}
        \left(\left(\alpha_{n+1} + \epsilon \left(v_{n+1,1}\alpha_{n+1}- 1\right)\right)\bm{v}_{n+1}
        + \sum_{i=1}^{n}\left(1+\epsilon \,v_{n+1,1}\right)\right)\alpha_{i}\bm{v}_{i}\right\|^2\\
        & = \left\|
        \left(\alpha_{n+1} + \epsilon \left(v_{n+1,1}\alpha_{n+1}- 1\right)\right)\lambda_{n+1}^{-1}\bm{v}_{n+1}
        + \sum_{i=1}^{n}\left(1+\epsilon \,v_{n+1,1}\right)\lambda_{i}^{-1}\alpha_{i}\bm{v}_{i}\right\|^2\\
        & = \left( \left(\alpha_{n+1} + \epsilon \left(v_{n+1,1}\alpha_{n+1}- 1\right)\right)\lambda_{n+1}^{-1}\right)^{2} \|\bm{v}_{n+1}\|^2
        + \sum_{i=1}^{n}\left( \left(1+\epsilon \,v_{n+1,1}\right)\lambda_{i}^{-1}\alpha_{i}\right)^{2} \|\bm{v}_{i}\|^2
        \\
        & \geq \left(\left(1 + \epsilon v_{n+1,1}\right) \alpha_{n+1}- \epsilon \right)^{2} \lambda_{n+1}^{-2}\|\bm{v}_{n+1}\|^2
        \\
        & \geq \epsilon^{2} \lambda_{n+1}^{-2}\,.
    \end{split}
\end{align}
Since $|\epsilon| \leq 1$ and then $1 + \epsilon v_{n+1,1} \geq 0$.
Together with \eqref{eq:epsilon},
we can obtain the second inequality that
\begin{align*}
    \begin{split}
        \left(\left(1 + \epsilon v_{n+1,1}\right) \alpha_{n+1}- \epsilon \right)^{2}
        \geq
        \left(1 + \epsilon v_{n+1,1}\right)^{2}\alpha_{n+1}^{2} + \epsilon^{2}
        \geq
        \epsilon^{2} \,.
    \end{split}
\end{align*}
Since $P_{t}$ is a symmetric positive definite matrix, we have
$\tr(P_{t}^{-1})^{-1} \leq \lambda_{\min}(P_{t}) \leq n \tr(P_{t}^{-1})^{-1}$.
Given \eqref{eq:traceLt} and by the Sherman-Morrison formula, then we have
\begin{equation}\label{eq:lambdaL}
    \lambda_{n+1} \leq (n+1)\tr(P_{t}^{-1})^{-1} \leq (n+1)\tr\left(\left(P_{t} + \bm{p}_{cep}\bm{p}_{cep}^{\top}\right)^{-1} \right)^{-1}
    < (n+1)L_{1}(t+1) \,.
\end{equation}
By \eqref{eq:normPtp} and \eqref{eq:lambdaL},
\begin{align}\label{divisor}
    \begin{split}
        \left\|P_{t}^{-1}\bm{p}{'}\right\|^2
        \geq \epsilon^{2} (n+1)^{-2}L_{1}({t+1})^{-2}
        \,.
    \end{split}
\end{align}
Together with \eqref{dividend} and \eqref{divisor}, we have
\begin{align*}
    \begin{split}
        \frac{\|P_{t}^{-1}\bm{p}{'}\|^2}{1 + \|\bm{p}{'}\|_{P_{t}^{-1}}^2}
        \geq \frac{\epsilon^{2} (n+1)^{-2}L_{1}({t+1})^{-2}}{1 + L_{1}(n+1)^{-1} \max_{\bm{p}\in \mathcal{P}}\|\bm{p}\|^{2}}
        \,.
    \end{split}
\end{align*}
Choose $\epsilon = K \sqrt{\dot{L}_{1}(t)}$
and given the LHS of the above inequality, we
let $$\kappa^{2} = K^{2} (n+1)^{-2}\left(1 + L_{1}(n+1)^{-1} \max_{\bm{p}\in \mathcal{P}}\|\bm{p}\|^{2}\right)^{-1}.$$
Here $\kappa \geq 1$ by \eqref{eq:LvsEpsilon}.
Then we obtain
\begin{align*}
    \begin{split}
        \frac{\|P_{t}^{-1}\bm{p}{'}\|^2}{1 + \|\bm{p}{'}\|_{P_{t}^{-1}}^2}
        \geq
        \frac{\kappa^{2}\dot{L}_{1}(t)}{L_{1}({t+1})^{2}}
        \,.
    \end{split}
\end{align*}
Due to the convexity of $\frac{1}{L_{1}(t)}$, there exists $\kappa > 1$, such that $\kappa L_{1}(t) \geq L_{1}(t+1)$, and
\begin{align*}
    \begin{split}
        \frac{\|P_{t}^{-1}\bm{p}{'}\|^2}{1 + \|\bm{p}{'}\|_{P_{t}^{-1}}^2}
        \geq \frac{\kappa^{2}{\dot{L}_{1}(t)}}{L_{1}({t+1})^{2}}
        \geq \frac{{\dot{L}_{1}(t)}}{L_{1}({t})^{2}}
        \,.
    \end{split}
\end{align*}
It shows that condition \eqref{eq:traceLt1} holds.
%%%%%%%%%%%%%%%%%%%%%%%%%%%%%%%%%%%%%%%%%%%%%%%%%
%% Proof of Proposition GLM Beta converge
%%%%%%%%%%%%%%%%%%%%%%%%%%%%%%%%%%%%%%%%%%%%%%%%%
\subsection{Proof of Proposition \ref{proposition_GLMBeta}}\label{Proof_PropGLMBeta}

%\Neil{ABOVE: Is this not already stated in the statement of Proposition 4.1. If it isn't we need to include it and if it is then we can remove this.}

The price vector can be written as a function in terms of $\bm{\beta}$ given by $\bm{p} = \bm{p}({\bm{\beta}})$.
Since $\bm{p}({\bm{\beta}_{0}}) \in \mathcal{P}$ and
$\frac{\partial r(\bm{p}^{*}, {\bm{\beta}_{0}})}{\partial {p}_{i}} = 0$ at the optimal price $\bm{p}^{*}$,
by the Taylor series expansion, we can derive
\begin{equation*}
    \left|r(\bm{p}, {\bm{\beta}_{0}})-r(\bm{p}^{*}, {\bm{\beta}_{0}})\right|
    \leq \frac{1}{2} \left(\text{sup}_{\bm{p} \in \mathcal{P}}\left|\frac{\partial^2 r(\bm{p}, {\bm{\beta}_{0}})}{ \partial {p}_{i}^2}\right|\right) \left\|\bm{p}-\bm{p}^{*}\right\|^2 \, ,
\end{equation*}
for all $\bm{p} \in \mathcal{P}$.
Let $k=\text{sup}_{\bm{p} \in \mathcal{P}}\left|\frac{\partial^2 r(\bm{p}, {\bm{\beta}_{0}})}{{p}_{i}^2}\right| < \infty$, we have
\begin{equation*}
    \left|r(\bm{p}, {\bm{\beta}_{0}})-r(\bm{p}^{*}, {\bm{\beta}_{0}})\right|
    \leq k \left\|\bm{p}-\bm{p}^{*}\right\|^2 \, .
\end{equation*}
By the definition of implicit function theorem in Duistermaat and Kolk \cite{DuistermaatKolk2004},
there is an open and bounded neighbourhood $V$ such that the function $\bm{\beta} \rightarrow \bm{p}(\bm{\beta})$ is continuously differentiable with bounded derivatives.
Thus, for all ${\bm{\beta}} \in V$ and some non-random constant $K_{2}> 0$, we have
\begin{equation*}
    \left\|\bm{p}({\bm{\beta}})-\bm{p}({{\bm{\beta}}}_{0})\left\| \leq K_{2} \right\| {\bm{\beta}} - {\bm{\beta}}_{0}\right\| \, .
\end{equation*}
Assume there exits $\widehat{\bm{\beta}}_{t} \in V$ for all $t$, then
\begin{equation*}
    \left\|\bm{p}({\widehat{\bm{\beta}}_{t}})-\bm{p}({\bm{\beta}}_{0})\right\|
    \leq K_{2} \left\| {\widehat{\bm{\beta}}_{t}} - {\bm{\beta}_{0}}\right\| \, .
\end{equation*}
This shows the upper bounds on the regret depend on the upper bounds on
$\left\|\widehat{\bm{\beta}}_{t}-{\bm{\beta}_{0}}\right\|^{2}$.
%We can derive several bounds on this quantity.

%% Proof of Proposition 4(i)
\subsection{Proof of Proposition \ref{proposition_GLMboundsClinkF}}\label{Proof_PropGLMboundsClinkF}
To prove Proposition \ref{proposition_GLMboundsClinkF}, we need the following lemma.
\begin{lemma}\label{ChenHfun_lemma}
Let $H$ be a smooth continuously differentiable injection from $\R^{d} \to \R^{d}$ with $H(x_{0}) = y_{0}$.
Define $B_{\rho} = \{x \in \R^{d}\,|\,\|x - x_{0}\|\leq \rho\}$ and
$\partial B_{\rho} = \{x \in \R^{d}\,|\,\|x - x_{0}\|= \rho\}$.
Then, $\inf_{x \in \partial B_{\rho}(x_{0})}\| H(x) - y_{0}) \| \geq \delta$ implies
$B_{\delta}(y_{0}) = \{y \in \R^{d}\,|\, \| y - y_{0} \| \leq \delta\} \subseteq H(B_{\rho}(x_{0}))$, which gives
$H^{-1}(B_{\delta}(y_{0})) \subseteq B_{\rho}(x_{0})$.
\end{lemma}

\begin{proof}
Since $H: B_{\rho}(x_{0}) \mapsto  H(B_{\rho}(x_{0}))$ is a homeomorphism, we can directly derive the result based on Theorem 3.1 \& Corollary 3.2 in Dugundji \cite{Dugundji1966}.
For any set $B \subset H^{d+1}$ and a homeomorphism $h: B \mapsto H^{d+1}$, if $x$ is an interior boundary point of $B$, then $h(x)$ is an interior boundary point of $h(B)$.
This is because by the Brouwer domain invariance theorem, for any space $B$, the property ``open in $H^{d+1}$" is a positional invariant of $B$ rel $H^{d+1}$.%Brouwer's fixed-point theorem
\end{proof}
It simply tells us that for all $y \in \{y \in \R^{d}\,|\, \| y - y_{0} \| \leq \delta\}$, there is an $x \in \{x \in \R^{d}\,|\, \| x - x_{0} \| \leq \rho \}$ such that $H(x) = y$.
Similarly, we define a closed and bounded neighbourhood of ${\bm{\beta}}_{0}$ as
$B_{\rho} = \left\{{\bm{\beta}} \in \mathbb{R}^{2} \, \big| \, \left\|{{\bm{\beta}}}-{\bm{\beta}}_{0}\right\| \leq \rho\right\}$ and
$\partial B_{\rho} = \left\{{\bm{\beta}} \in \mathbb{R}^{2} \, \big| \, \left\|{{\bm{\beta}}}-{\bm{\beta}}_{0}\right\|= \rho \right\}$.

Now, we shall start our proof of Proposition \ref{proposition_GLMboundsClinkF}, that is to find the value of $\Exp \left[ \left\|\widehat{\bm{\beta}}_{t} - {\bm{\beta}}_{0}\right\|^2 \right]$.
\begin{proof}[Proof of Proposition \ref{proposition_GLMboundsClinkF}]
Since $\bm{\beta}$ is not effected by $\sigma$, without loss of generality we assume that $\sigma = 1$.
A Taylor expansion of $h(\cdot)$ yields
\begin{align*}
    \begin{split}
        y_{i} - h(\bm{p}({i})^{\top}\bm{\beta})
        & = y_{i} - h(\bm{p}({i})^{\top}\bm{\beta}_{0}) + h(\bm{p}({i})^{\top}\bm{\beta}_{0}) - h(\bm{p}({i})^{\top}\bm{\beta})  \\
        & = y_i - h(\bm{p}({i})^{\top}{\bm{\beta}}_{0}) - \dot{h}(\bm{p}({i})^{\top}\widetilde{\bm{\beta}})\bm{p}({i})^{\top}(\bm{\beta}-\bm{\beta}_{0}) \, ,
    \end{split}
\end{align*}
for some $\widetilde{\bm{\beta}}$ is on the line segment between $\bm{\beta}_{0}$ and $\bm{\beta}$.
Thus
\begin{align*}
    \begin{split}
        {\bm{l_{t}}}(\bm{\beta}) - {\bm{l_{t}}}(\bm{\beta}_{0})
        & = \sum_{i=1}^{t} \bm{p}({i}) (h(\bm{p}({i})^{\top}){\bm{\beta}}_{0} - h(\bm{p}({i})^{\top}){\bm{\beta}}) \\
        & = \sum_{i=1}^{t} \bm{p}({i}) \bm{p}({i})^{\top} \dot{h}(\bm{p}({i})^{\top}\widetilde{{\bm{\beta}}})({\bm{\beta}}_{0} - {\bm{\beta}}) \,.
    \end{split}
\end{align*}
Under the Assumptions A\ref{Assum_A1} and A\ref{Assum_A2}, the strong consistency for $\widehat{\bm{\beta}}_{t}$ in Lai and Wei \cite{LaiWei1982} holds.
We define
\begin{align*}
    \begin{split}
        E_{t}
        & = \left(\sum_{i=1}^{t} \bm{p}({i}) \bm{p}({i})^{\top} \right)^{-1}
        \sum_{i=1}^{t} \bm{p}({i}) \eta_{i}
        = P_{t}^{-1} \sum_{i=1}^{t} \bm{p}({i}) \eta_{i} \,,
    \end{split}
\end{align*}
here
$\eta_{i}:=y_i - h(\bm{p}({i})^{\top}{\bm{\beta}}_{0})$.
Since $\Exp[\eta_{i} \,| \,\mathcal{F}_{i-1}] = 0$,  $\|E_{t}\|\to 0$ a.s.
Write $H_{t}(\bm{\beta}) = P_{t}^{-1} l_{t}(\bm{\beta})$ and
let $\kappa = \inf_{\bm{\beta} \in B_{\rho}, \bm{p} \in \mathcal{P}}\dot{h}(\bm{p}({i})^{\top}{\bm{\beta}})>0 $.
If $\left\|{\bm{\beta}}-{\bm{\beta}}_{0}\right\| \leq \rho $,
then for all $\bm{\beta} \in \ B_{\rho}$, we have
\begin{align*}
    \begin{split}
        \left\|H_{t}(\bm{\beta}) - H_{t}({\bm{\beta}}_{0})\right\|^2
        & = \left\|P_{t}^{-1} \left(\bm{l}_{t}(\bm{\beta}) - \bm{l}_{t}(\bm{\beta}_{0})\right)\right\|^2 \\
        & = ({\bm{\beta}}_{0}-{\bm{\beta}})^{\top} \sum_{i=1}^{t} \bm{p}({i}) \bm{p}({i})^{\top} \dot{h}(\bm{p}({i})^{\top}\widetilde{\bm{\beta}}) \left(\sum_{i=1}^{t}\bm{p}({i}) \bm{p}({i})^{\top}\right)^{-2} \sum_{i=1}^{t} \bm{p}({i}) \bm{p}({i})^{\top} \dot{h}(\bm{p}({i})^{\top}\widetilde{\bm{\beta}}) ({\bm{\beta}}_{0}-{\bm{\beta}}) \\
        & \geq ({\bm{\beta}}_{0}-{\bm{\beta}})^{\top} \sum_{i=1}^{t} \bm{p}({i}) \bm{p}({i})^{\top} \dot{h}(\bm{p}({i})^{\top}\widetilde{\bm{\beta}}) \left(\sum_{i=1}^{t}\frac{\dot{h}(\bm{p}({i})^{\top}\widetilde{\bm{\beta}})}{\kappa}
        \bm{p}({i})\bm{p}({i})^{\top}\right)^{-2} \\
        & \sum_{i=1}^{t} \bm{p}({i}) \bm{p}({i})^{\top} \dot{h}(\bm{p}({i})^{\top}\widetilde{\bm{\beta}}) ({\bm{\beta}}_{0}-{\bm{\beta}}) \\
        & = \kappa^{2} \left\|{\bm{\beta}}_{0}-{\bm{\beta}} \right\|^{2} \, .
    \end{split}
\end{align*}
In particular, if $\left\|{\bm{\beta}}_{0}-{\bm{\beta}} \right\| = \rho$, we have
\begin{equation*}
    \left\|H_{t}(\bm{\beta}) - H_{t}({\bm{\beta}}_{0})\right\|^2 \geq \kappa^{2} \rho^{2} > 0  \,.
\end{equation*}
Let $ \widetilde{H}_{t}(\bm{\beta}) = H_{t}(\bm{\beta}) - H_{t}({\bm{\beta}}_{0})$ with $\widetilde{H}_{t}(\bm{\beta}_{0}) = 0$, we have $\Exp \left[ \left\|{\bm{\beta}}_{t} - {\bm{\beta}}_{0}\right\|^2 \right]$ = $O\left(\Exp \left[ \left\| \widetilde{H}_{t}({\bm{\beta}}_{t})\right\|^2  \right]\right)$.
By Lemma \ref{ChenHfun_lemma}, we know $\widetilde{H}_{t}^{-1}(E_{t})$ is well defined on $\{\bm{\beta}:\|\bm{\beta} - \bm{\beta}_{0}\|\leq \rho \}$.
Since
\begin{equation*}
    \sum_{i=1}^{t} \bm{p}({i}) \left(h(\bm{p}({i})^{\top}\widetilde{H}_{t}^{-1}(E_{t})) - y_{i}\right)
    =
    \sum_{i=1}^{t} \bm{p}({i})\bm{p}({i})^{\top} \left(\widetilde{H}_{t}\left(\widetilde{H}_{t}^{-1}(E_{t}))\right) - E_{t}\right)
    = 0 \,,
\end{equation*}
we obtain that
$\widehat{\bm{\beta}}_{t} = \widetilde{H}_{t}^{-1}(E_{t})$ exists and $\widetilde{H}_{t}(\widehat{\bm{\beta}}_{t}) = E_{t}$.
Given $\lambda_{\min}(t) \to \infty$ as $t \to \infty$ and ${P}_{t}$ is nonsingular for all large $t$, we have,
\begin{align*}
    \begin{split}
        \left\| \widetilde{H}_{t}(\widehat{\bm{\beta}}_{t})\right\|^2
        & =  \left\| P_{t}^{-1} \sum_{i=1}^{t} \bm{p}({i}) \eta_{i} \right\|^2
        \leq \Bigg\| P_{t}^{-\frac{1}{2}}\Bigg\|^2 \left\| P_{t}^{-\frac{1}{2}}   \sum_{i=1}^{t} \bm{p}({i}) \eta_{i} \right\|^2 \\
        & = \lambda_{\min}^{-1}(t) \left(\sum_{i=1}^{t}\bm{p}({i})^{\top}\eta_{i}\right) P_{t}^{-1} \left(\sum_{i=1}^{t}\bm{p}({i})\eta_{i}\right)\,.
    \end{split}
\end{align*}
Define $N=\inf\{t:P_{t}\text{ is nonsigular}\}$.
Assume $N < \infty$ and for $t \geq N$, we
let $V_{t} = P_{t}^{-1}$ and
\begin{equation*}
    Q_{t}
    = \left(\sum_{i=1}^{t}\bm{p}({i})^{\top}\eta_{i}\right) V_{t} \left(\sum_{i=1}^{t}\bm{p}({i})\eta_{i}\right)\,.
\end{equation*}
%Since $P_{t} = P_{t-1} + \bm{p}(t)\bm{p}(t)^{\top}$ and by the Sherman-Morrison formula, we obtain
%$$P_{t}^{-1} = P_{t-1}^{-1} - \frac{P_{t-1}^{-1}p_{t}p_{t}^{\top} P_{t-1}^{-1}}{1 + p_{t}^{\top} P_{t}^{-1}p_{t}} \,.$$
Apply the Sherman-Morrison formula, we obtain the recursive form of $Q_{t}$.
For $k>N$, we have
\begin{align*}
    \begin{split}
        Q_{k}
        & = \left(\sum_{i=1}^{k}\bm{p}({i})^{\top}\eta_{i}\right) V_{k} \left(\sum_{i=1}^{k}\bm{p}({i})\eta_{i}\right) \\
        & = \left(\sum_{i=1}^{k-1}\bm{p}({i})^{\top}\eta_{i}\right)V_{k} \left(\sum_{i=1}^{k-1}\bm{p}({i})\eta_{i}\right)
        + \bm{p}({k})^{\top} V_{k} \bm{p}({k}) \eta_{k}^{2}
        + 2 \bm{p}({k})^{\top} V_{k} \left(\sum_{i=1}^{k-1}\bm{p}({i})\eta_{i}\right) \eta_{k} \\
        & = Q_{k-1} + \left(\sum_{i=1}^{k-1}\bm{p}({i})^{\top}\eta_{i}\right)
        \left(-\frac{V_{k-1}\bm{p}({k})\bm{p}({k})^{\top}V_{k-1}}{1 + \bm{p}({k})^{\top}V_{k-1}\bm{p}({k})}\right)
        \left(\sum_{i=1}^{k-1}\bm{p}({i})\eta_{i}\right)  \\
        & + \bm{p}({k})^{\top} V_{k} \bm{p}({k}) \eta_{k}^{2}
        +  2 \left(\frac{\bm{p}({k})^{\top}V_{k-1}}{1 + \bm{p}({k})^{\top}V_{k-1}\bm{p}({k})}\right)
        \left(\sum_{i=1}^{k-1}\bm{p}({i})\eta_{i}\right) \eta_{k} \\
        & = Q_{k-1}
        - \frac{\left( \bm{p}({k})^{\top}V_{k-1}\sum_{i=1}^{k-1}\bm{p}({i})\eta_{i}\right)^{2}}{1 + \bm{p}({k})^{\top}V_{k-1}\bm{p}({k})} + \bm{p}({k})^{\top} V_{k} \bm{p}({k}) \eta_{k}^{2}
        +  2 \left(\frac{\bm{p}({k})^{\top}V_{k-1}}{1 + \bm{p}({k})^{\top}V_{k-1}\bm{p}({k})}\right)
        \left(\sum_{i=1}^{k-1}\bm{p}({i})\eta_{i}\right) \eta_{k} \, .
    \end{split}
\end{align*}
Let %\cite{Pena2009}Page219&220-(squared) Studentized statistic
\begin{equation*}
    \gamma_{k} = \frac{\left( \bm{p}({k})^{\top}V_{k-1}\sum_{i=1}^{k-1}\bm{p}({i})\eta_{i}\right)^{2}}{1 + \bm{p}({k})^{\top}V_{k-1}\bm{p}({k})}\,, \quad
    \theta_{k} = \bm{p}({k})^{\top} V_{k} \bm{p}({k}) \eta_{k}^{2} \,, \quad
    \omega_{k-1} = 2 \left(\frac{\bm{p}({k})^{\top}V_{k-1}}{1 + \bm{p}({k})^{\top}V_{k-1}\bm{p}({k})}\right)
    \left(\sum_{i=1}^{k-1}\bm{p}({i})\eta_{i}\right)\,.
\end{equation*}
Here, $\gamma_{k}, \theta_{k}\geq 0$ and $\omega_{k-1}$ are $\mathcal{F}_{k-1}$-measurable.
Summing it, we have for $t > N$,
\begin{align*}
    \begin{split}
        Q_{t}
        & = Q_{N}
        - \sum_{k=N+1}^{t}
        \gamma_{k}
        + \sum_{k=N+1}^{t} \theta_{k}
        + \sum_{k=N+1}^{t}  \omega_{k-1}\eta_{k}
         \, .
    \end{split}
\end{align*}
Here, $Q_{t} \geq 0$ is an extended stochastic Liapounov function if it is $\mathcal{F}_{t}$-measurable (Lai \cite{Lai2003}).
By the strong laws for martingale, for any $\alpha>0$, we have
\begin{align*}
    \begin{split}
        \max \left(Q_{t}, \sum_{k=N+1}^{t}
        \gamma_{k} \right)
        = O \left(\sum_{k=N+1}^{t} \theta_{k} + \left(\sum_{k=N+1}^{t} \omega_{k-1}^{2}\right)^{\frac{1}{2}+\alpha}\right)
        \,.
    \end{split}
\end{align*}
The local martingale convergence theorem and the strong law of large numbers (Chow \cite{Chow1965}) show
$\sum_{k=N+1}^{t}  \omega_{k-1}^{2}\leq 4 \sum_{k=N+1}^{t}\gamma_{k}$.
When $\lim_{t \to \infty} \lambda_{\max}(t) = \infty$, by Kronecker's lemma and Freeman theorem \cite{Freedman1973}, we have
\begin{equation*}
    \sum_{k=N+1}^{t} \theta_{k}
    = O \left(\sum_{k=N+1}^{t} \bm{p}({k})^{\top}  V_{k} \bm{p}({k})\right)
    = O (\log \lambda_{\max}(t))\,.
\end{equation*}
It implies when $\lim_{t \to \infty} \lambda_{\max}(t) = \infty$, we have
\begin{equation*}
    Q_{t}
    = O \left(\log \lambda_{\max} (t)\right)
    \quad \text{ and } \quad
    \sum_{k=N+1}^{t} \gamma_{k} = O \left(\log \lambda_{\max} (t)\right) \, .
\end{equation*}
Assume Assumption A\ref{Assum_A1} and A\ref{Assum_A2} hold,
%\begin{equation*}\lim_{n \to \infty} \frac{\lambda_{\min}(t)}{\log(t)} = \infty \,,\end{equation*}
by Theorem 2 in Lai and Wei \cite{LaiWei1982} we obtain that $\widehat{\bm{\beta}}_{t}$ is strongly consistent with
\begin{equation*}
    \left\|\widehat{\bm{\beta}}_{t} - {\bm{\beta}}_{0}\right\|^2
    = O \left(\frac{\log \lambda_{\max} (t)}{\lambda_{\min} (t)}\right)
    = O\left(\frac{\log(t)}{\lambda_{\min}(t)}\right) \,.
\end{equation*}
Assume that $\lambda_{\min}(t) \geq L_{1}(t)$ holds, we have
\begin{equation*}
    \Exp\left[\left\|\widehat{\bm{\beta}}_{t}-{\bm{\beta}}_{0}\right\|^2 \right]
    = O \left(\frac{\log(t)}{L_{1}(t)}\right) \, .
\end{equation*}
\end{proof}
%%%%%%%%%%%%%%%%%%%%%%%%%%%%%%%%%%%%%%%%%%%%%%%%%
%% Proof of Proposition delay case
%%%%%%%%%%%%%%%%%%%%%%%%%%%%%%%%%%%%%%%%%%%%%%%%%

%%%%%%%%%%%%%%%%%%%%%%%%%%%%%%%%%%%%%%%%%%%%%%%%%
%% Proof of Theorem GP regret bounds
%%%%%%%%%%%%%%%%%%%%%%%%%%%%%%%%%%%%%%%%%%%%%%%%%
\subsection{Proof of Theorem \ref{theorem_GPBounds}}\label{Appendix_GPUCB:RegBounds}
In this section, we present lemmas used for the proof of Theorem \ref{theorem_GPBounds}.
This proof is based on the work of Srinivas et al.\@ \cite{Srinivas2010}.
We first present the following results from \cite{Srinivas2010}.

\begin{lemma} \cite[Lemma 5.3]{Srinivas2010}\label{Lemma:InfoGain}
If $f_{T} = (f({p}_{t})) \in \R^{T}$,
the information gain in GP can be expressed as
\begin{equation*}
    I\left(\bm{y}_{T}; f_{T}\right)
    = \frac{1}{2}\sum_{t=1}^{T}
    \log\left(1+{\sigma}^{-2}{\sigma}_{t-1}({p_{t}})^{2}\right) \,.
\end{equation*}
Here, ${\sigma}^{2}$ is the variance of Gaussian noise and ${\sigma}_{t-1}({p_{t}})^{2}$ is the posterior variance after $t-1$ observations.
\end{lemma}
\begin{proof}
The Shannon Mutual Information $I$ is defined as
\begin{equation}\label{eq:defI}
    I\left(\bm{y}_{T}; f_{T}\right) = H(\bm{y}_{T}) - H(\bm{y}_{T}|f) \,.
\end{equation}
It quantifies the reduction in uncertainty (measured in terms of differential Shannon entropy) about $f$ from revealing $\bm{y}$.
%the difference between the entropy of the prior distribution and the entropy of the posterior, given by
By the definition of $H(\cdot)$, we have that for a Gaussian Process,
\begin{align}\label{eq:defH}
    \begin{split}
        H(\bm{y}_{T}|f)
        & = \frac{1}{2} \log |2\pi e \sigma^2 \bm{I}|=\frac{1}{2} \sum_{t=1}^{T} \log (2\pi e \sigma^2 )\,.
    \end{split}
\end{align}
We can expand $H(\bm{y}_{T})$ as
\begin{align*}
    \begin{split}
        H(\bm{y}_{T})
        & = H(\bm{y}_{T-1}) + H(\bm{y}_{T}|\bm{y}_{T-1})\\
        & = H(\bm{y}_{T-1}) + \frac{1}{2} \log(2\pi e (\sigma^2 +{\sigma}_{t-1}({p_{t}})^{2}))\,.
    \end{split}
\end{align*}
By expanding the entropy terms, we have
\begin{align}\label{eq:expdH}
    \begin{split}
        H(\bm{y}_{T})
        & = \frac{1}{2} \sum_{t=1}^{T}\log(2\pi e (\sigma^2 +{\sigma}_{t-1}({p_{t}})^{2}))\,.
    \end{split}
\end{align}
Substituting \eqref{eq:defH} and \eqref{eq:expdH} into \eqref{eq:defI}, we
can obtain the result.
\end{proof}

The following lemma is used to obtain the finite bound on $\gamma_{T}$.
The proof is quite involved and we refer the reader to Theorem 8 in Srinivas et al.\@ \cite{Srinivas2010}.
\begin{lemma} \cite[Theorem 8]{Srinivas2010}\label{InfoGainBound}
Assume that Assumption \ref{Ass:f} holds.
Choose $n_{T} = c T^{\tau} \log T$, where $c$ is a constant and $\tau > 0$.
For any $T_{*} \in \{1, \dots, \min(T, n_{T})\}$,
we let $B_{k}(T_{*}) = \sum_{s>T_{*}}\lambda_{s}$, here $\lambda_{s}$ is the eigenvalues of kernel $k({p}, {p'})$ w.r.t the uniform distribution over $\mathcal{P}$,
Then the bounds on $\gamma_{T}$ is given by
\begin{align*}
    \begin{split}
        \gamma_{T}
        & \leq \frac{1/2}{1-e^{-1}}
        \max_{r \in \{1, \dots, {T}\}}
        ( T_{*} \log(rn_{T}/\sigma^{2}) \\
        & + c \sigma^{2} (1 - r/T) (\log T)(T^{t+1}B_{k}(T_{*})+1) ) \\
        & + O(T^{1-t/d})\,.
    \end{split}
\end{align*}
\end{lemma}
We also need to establish the following lemmas before we prove Theorem \ref{theorem_GPBounds}.
Lemma \ref{lemma:Confbound1} provides a confidence bound on a finite decision set $|\mathcal{P}| < \infty$, where all decisions are chosen.
Lemma \ref{lemma:Confbound2} shows a confidence bound on a set of discretizations $\mathcal{P}_{t} \subset \mathcal{P}$ where $\mathcal{P} \subset \R^{d}$ is a general compact set.

Denote a sequence of $\pi_{t} > 0$ such that $\sum_{t}\pi_{t}^{-1} = 1$.
\begin{lemma}\label{lemma:Confbound1}
Pick $\delta \in (0,1)$ and set $\varphi_{t} = 2 \log(|\mathcal{P}|\pi_{t} / \delta)$.
Then with probability greater than $ 1-\delta$, for any $p \in \mathcal{P}$ and $t \geq 1$, we have
\begin{equation}\label{eq:fmu}
    \left| r({p}) - \mu_{t-1}^{r}({p}) \right| \leq \sqrt{\varphi_{t}} {\sigma}_{t-1}^{r}({p}_{t})\,.
\end{equation}
Here, ${\sigma}_{t-1}^{r}({p}) = p_{t-1}{\sigma}_{t-1}^{d}({p}) + {\sigma}_{t-1}^{c}({p})$.
\end{lemma}
\begin{proof}
Conditioned on  $\bm{y}_{t-1} = ({y}_{1}, \dots, {y}_{t-1})$, $\{{p}_{1}, \dots, {p}_{t-1}\}$ are deterministic and the marginals follow $f({p}) \sim \mathcal{N}(\mu_{t-1}({p}), \sigma_{t-1}^{2}({p}))$ for any fixed $t \geq 1$ and ${p} \in \mathcal{P}$ .
By the following tail bound, we know $\Pro(z>c) \leq (1/2)e^{-c^2 /2}$ for $c > 0$ if $z\sim\mathcal{N}(0, 1)$.
Let $z = \left({f({p}) - \mu_{t-1}({p})}\right) /{\sigma_{t-1}({p})}$ and
$c = \varphi_{t}^{1/2}$, then
\begin{equation*}
    \Pro\left(\frac{\left| f({p}) - \mu_{t-1}({p}) \right|}{\sigma_{t-1}({p})} > \varphi_{t}^{1/2} \right)
    \leq e^{-\varphi_{t} /2} \,.
\end{equation*}
With probability greater than $ 1 - |\mathcal{P}|e^{-\varphi_{t} /2}$, we have
\begin{equation*}
    {\left| f({p}) - \mu_{t-1}({p}) \right|}\leq \varphi_{t}^{1/2} {\sigma_{t-1}({p})}\,.
\end{equation*}
Given $r(p) = p \cdot f_{d}(p) - f_{c}(p)$, with probability greater than $ 1 - |\mathcal{P}|e^{-\varphi_{t} /2}$, we have
\begin{align*}
    \begin{split}
         {\left| r({p}) - \mu_{t-1}^r({p}) \right|}
         & = \left| \left(p f_{d}(p) - f_{c}(p)\right) - \left( p \mu^{d} - \mu^{c}\right) \right| \\
         & \leq \left| \left(p f_{d}(p) - p \mu^{d} \right)\right| + \left|\left(  f_{c}(p)- \mu^{c}\right) \right|\\
         & \leq p \varphi_{t}^{1/2}  {\sigma_{t-1}^{d}({p})} + \varphi_{t}^{1/2} {\sigma_{t-1}^{c}({p})}\,.
    \end{split}
\end{align*}
Let ${\sigma}_{t-1}^{r}({p}) = p_{t-1}{\sigma}_{t-1}^{d}({p}) + {\sigma}_{t-1}^{c}({p})$ and choose $|\mathcal{P}|e^{-\varphi_{t} /2} = \delta/\pi_{t}$.
%For all $t \geq 1$, we let $|\mathcal{P}|e^{-\varphi_{t} /2} = \delta/\pi_{t}$ and choose $\pi_{t} = \pi^2 t^2/6$ since $\sum 1/t^2 = \pi^2/6$.
By the union bound on all $t$, we obtain the results.
\end{proof}

By the Assumption \ref{Ass:f} and the union bound we have $\Pro \left(\sup_{{p} \in \mathcal{P}} \left|{\partial f }/{\partial p_{j}}\right| > J \right) \leq a' e^{-(J/b')^2}$.
Therefore, there exits $c>0$, for $a, b > 0$, such that
\begin{equation*}
    \Pro \left(\sup_{p \in \mathcal{P}} \left| \frac{\partial r}{\partial p_{j}}\right| > J \right) \leq a e^{-(J/b)^2}\,.
\end{equation*}
Then with probability greater than $ 1 - a e^{-(J/b)^2}$, for all $p \in \mathcal{P}$, we have
\begin{equation}\label{eq:f-f}
    \left| r(p) - r(p')\right|
    \leq
    J \left\|p - p'\right\| \,.
\end{equation}
Now, consider a sequence of discretisation $\mathcal{P}_{t}$ of cardinality $\zeta_{t}$ that satisfies
\begin{equation} \label{eq:xbound}
    \left\|p - [p]_{t}\right\| \leq r /\zeta_{t} \,,
\end{equation}
where $r = p_{h} - p_{l}$ is a constant which is the length of price set and $[p]_{t}$ is the closest price to $p$ in $\mathcal{P}_{t}$.
Let $\delta/2 = a e^{-(J/b)^2}$ and apply \eqref{eq:f-f}.
With probability greater than $ 1 - \delta/2$,
\begin{equation*}
    \left| r(p) - r(p')\right|
    \leq
    b\sqrt{\log(2a/\delta)} \left\|p - p'\right\| \,.
\end{equation*}
Together with \eqref{eq:xbound}, we can derive
\begin{align*}
    \begin{split}
        \left| r(p) - r([p]_{t})\right|
        & \leq
        b\sqrt{\log(2a/\delta)} \left\|p - [p]_{t}\right\| \\
        & \leq
        rb\sqrt{\log(2a/\delta)} /\zeta_{t} \,.
    \end{split}
\end{align*}
Choosing $\zeta_{t} =t^2$ yields
\begin{align*}
    \left| r(p) - r([p]_{t})\right| \leq {r b \sqrt{\log(2a/\delta)}}/{t^2}\,.
\end{align*}

We can now derive the bounds on regret.
\begin{lemma}\label{lemma:Confbound2}
Pick $\delta \in (0,1)$ and set
$$
\varphi_{t} = 2 \log( 2\pi_{t}^2 t^2 / 3\delta)
+ 2 \log\left(t^2 \right) \,.
$$
Then with probability greater than $ 1-\delta/2$, for any $p \in \mathcal{P}$ and $t \geq 1$, we have
\begin{equation}\label{eq:f[mu]}
    \left| r({p}) - \mu_{t-1}^{r}([{p}]_{t}) \right| \leq \varphi_{t}^{1/2} \sigma_{t-1}^{r}([{p}]_{t}) + \frac{r b  \sqrt{\log(2a/\delta)}}{t^2}\,.
\end{equation}
\end{lemma}
\begin{proof}

Together with \eqref{eq:xbound}, we can derive
\begin{align*}
    \begin{split}
        \left| r(p) - r([p]_{t})\right|
        & \leq
        b\sqrt{\log(2a/\delta)} \left\|p - [p]_{t}\right\| \\
        & \leq
        rb\sqrt{\log(2a/\delta)} /\zeta_{t} \,.
    \end{split}
\end{align*}
By Lemma \ref{lemma:Confbound1}, for all $p \in \mathcal{P}$ with probability greater than $ 1 - \delta/2$,
\begin{align*}
    \left| r({p}) - \mu_{t-1}^{r}([{p}]_{t}) \right|
    & \leq
    \left| r({p}) - r([p]_{t}) \right| +
    \left| r([p]_{t}) -  \mu_{t-1}^{r}([{p}]_{t})\right| \\
    & \leq {r b \sqrt{\log(2a/\delta)}}/{t^2} + \varphi_{t}^{1/2} \sigma_{t-1}^{r}([{p}]_{t})\,.
\end{align*}
Then obtain the results.
\end{proof}

\begin{lemma}\label{lemma:simpleregretbound2}
with probability greater than $ 1 - \delta$,  for all $t \geq 1$, the regret is bounded by
\begin{align*}
    r(p^{*}) -  r(p_{t}) \leq 2\sqrt{\varphi_{t}} {\sigma}_{t-1}^{r}({p}_{t})+  \frac{r b  \sqrt{\log(2a/\delta)}}{t^2}\,.
\end{align*}
\end{lemma}

\begin{proof}
By the definition that ${p}_{t} = \argmax_{{p} \in \mathcal{P}}{{\mu}_{t-1}^{r}({p})+ \sqrt{\varphi_{t}} {\sigma}_{t-1}^{r}({p})}$, we have
\begin{equation}\label{eq:ucb1}
    {\mu}_{t-1}^{r}({p}_{t})+ \sqrt{\varphi_{t}} {\sigma}_{t-1}^{r}({p}_{t}) \geq {\mu}_{t-1}^{r}([{p}^*]_{t})+ \sqrt{\varphi_{t}} {\sigma}_{t-1}^{r}([{p}^*]_{t}) \,.
\end{equation}
By \eqref{eq:f[mu]}, we have
\begin{equation}\label{eq:ucb2}
    {\mu}_{t-1}^{r}([{p}^*]_{t})+ \sqrt{\varphi_{t}} {\sigma}_{t-1}^{r}([{p}^*]_{t}) \geq r({p}^*) - r b\sqrt{\log(2a/\delta)} /t^2 \,.
\end{equation}
\eqref{eq:ucb1} and \eqref{eq:ucb2} imply
\begin{equation*}
    {\mu}_{t-1}^{r}({p}_{t})+ \sqrt{\varphi_{t}}
    {\sigma}_{t-1}^{r}({p}_{t}) \geq r({p}^*) - r b\sqrt{\log(2a/\delta)} /t^2 \,.
\end{equation*}
Together with %the confidence bound
\eqref{eq:fmu}, we can derive
\begin{align*}
    \begin{split}
        r(p^{*}) - r(p_{t})
        & \leq {\mu}_{t-1}^{r}({p}_{t})+ \sqrt{\varphi_{t}} {\sigma}_{t-1}^{r}({p}_{t}) + \frac{r b  \sqrt{\log(2a/\delta)}}{t^2} - r({p}_{t}) \\
        & \leq 2\sqrt{\varphi_{t}} {\sigma}_{t-1}^{r}({p}_{t})+  \frac{r b  \sqrt{\log(2a/\delta)}}{t^2}\,.
    \end{split}
\end{align*}
\end{proof}

\begin{lemma}\label{lemma:InfoGain2kernels}
For the combination of additive kernels $k_{d}({p},  {p}') + k_{c}({p},  {p}')$, we have
\begin{equation*}
    \gamma_{T}(k_{d}+k_{c})
    \leq \gamma_{T}(k_{d}) + \gamma_{T}(k_{c}) \,.
\end{equation*}
\end{lemma}
\begin{proof}
For $\mathcal{A} \subset \mathcal{P}$, let $\bm{K}_{d}$ be the Gram matrix on $\mathcal{A}$ for $k_{d}({p},  {p}')$ and $\bm{K}_{c}$ be the Gram matrix for $k_{c}({p},  {p}')$.
%By Madiman \cite{Madiman2008},
We can show that
\begin{align*}
    \begin{split}
        \log \left|\bm{I} + \sigma^{-2} \left(\bm{K}_{d}+\bm{K}_{c}\right)\right|
        \leq
        \log \left|\bm{I} + \sigma^{-2} \bm{K}_{d}\right| + \log \left|\bm{I} + \sigma^{-2} \bm{K}_{c}\right| \,.
    \end{split}
\end{align*}
It gives that
\begin{equation*}
    \gamma_{T}(k_{d}+k_{c})
    \leq \gamma_{T}(k_{d}) + \gamma_{T}(k_{c}) \,.
\end{equation*}
\end{proof}
%%%%%%%%%%%%%%%%%%%%%%%%%%%%%%%%%%%%%%%%%%%%%%%%%
%% Proof of GLM delay case
%%%%%%%%%%%%%%%%%%%%%%%%%%%%%%%%%%%%%%%%%%%%%%%%%
\subsection{Proof of Theorem \ref{theorem_GLMDelayBound}}\label{Appendix_DelayGLM}
To prove Theorem \ref{theorem_GLMDelayBound}, we need following lemmas.
\begin{lemma}\label{lemma:sumtildetau}
$\sum_{s=1}^{T}\widetilde{\tau}_{s} = \sum_{t=1}^{T}{\tau}_{t}$ \,.
\end{lemma}
\begin{proof}
By the definition of $\widetilde{\tau}_{s}$, we have
\begin{align*}
    \begin{split}
        \sum_{s=1}^{T}\widetilde{\tau}_{s}
        & = \sum_{s=1}^{T} \left(s-1-N(\rho(s))\right)\\
        & = \sum_{t=1}^{T} \left(t-1\right)-\sum_{s=1}^{T} N(\rho(s))\\
        & = \sum_{t=1}^{T} \left(t-1\right)-\sum_{t=1}^{T} N(t)\\
        & = \sum_{t=1}^{T} \left(t-1-N(t)\right)\\
        & = \sum_{t=1}^{T} \tau_{t} \,.
    \end{split}
\end{align*}
The third equality is obtained because $\{\rho(s) : s=1, \dots, T\}$ is a permutation of $\{1, \dots, T\}$.
The forth equality is derived by the definition of $N(t) = \sum_{i=1}^{t-1} \mathbbm{1}\{{i + \tau_{i} \leq t-1}\}$.
\end{proof}

\begin{lemma}\label{lemma:tildetau}
Given a maximum waiting time $m$ and ${\tau}_{t} \leq m$ for all $t$, we have $\widetilde{\tau}_{s} \leq 2m$ for all $s$.
\end{lemma}
\begin{proof}
By the definition of $\widetilde{\tau}_{s}$, we have
\begin{align*}
    \begin{split}
        \widetilde{\tau}_{s}
        & = s - 1 - N(\rho(s))\\
        & \leq s - 1 - \left(\rho(s) - 1- m \right)\\
        & \leq \tau_{\rho(s)} + m \\
        & \leq 2m \,.
    \end{split}
\end{align*}
In the first inequality, we use $N(\rho(s))  \geq \rho(s) - 1 - m$.
This is because at the beginning of time $\rho(s)$, there are at least $\rho(s) - 1 - m$ and at most $\rho(s)-1$ claims can be observed.
In the second inequality, we assume that the $s$-th claim incurs when a price is chosen at time $\rho(s)$ and then observed at time $\rho(s) + \tau_{\rho(s)}$.
The number of claims that observed by time $\rho(s) + \tau_{\rho(s)}-1$ is $N(\rho(s) + \tau_{\rho(s)})$.
Then the first claim that observed in time $\rho(s) + \tau_{\rho(s)}$ is $N(\rho(s) + \tau_{\rho(s)}) + 1$, which is also the $s$-th claim.
It implies
\begin{align*}
    \begin{split}
        s = N(\rho(s) + \tau_{\rho(s)}) + 1
        \leq \rho(s) + \tau_{\rho(s)}\,.
    \end{split}
\end{align*}
Since ${\tau}_{t} \leq m$ for all $t$, we obtain $\widetilde{\tau}_{s} \leq 2m$.
\end{proof}

%%%%%%%%%%%%%%%%%%%%%%%%%%%%%%%%%%%%%%%%%%%%% References %%
%%%%%%%%%%%%%%%%%%%%%%%%%%%%%%%%%%%%%%%%%%%
\section*{References}

\bibliography{paper}

\end{document}